\documentclass[11pt]{article} 
\usepackage{geometry}
\usepackage{authblk}
\usepackage{amsfonts, amsmath, amsthm, amssymb, enumerate, rotate}
\usepackage{amssymb, bm, bbm, latexsym, mathrsfs}
\usepackage{graphicx}% Required for inserting images
\usepackage{color, url}
\usepackage{times}
\usepackage[english]{babel}
\newtheorem{theorem}{Theorem}[section]

\usepackage{algorithm}
\usepackage{algorithmic}
\usepackage[numbers]{natbib}
\usepackage[title]{appendix}
\usepackage{url}
\usepackage{epstopdf}
\usepackage{paralist} % used for: compactitem
\usepackage{float}

\begin{document}

\title{Predictive Power Analysis of \\Multiple Test Procedures Under Arbitrary Dependence}
\author[1]{George Karabatsos}
\affil[1]{Departments of Mathematics, Statistics, and Computer Sciences\protect\\
and Educational Statistics \protect
\\ University of Illinois-Chicago \protect\\
e-mail: gkarabatsos1@gmail.com, georgek@uic.edu}
\date{\today}

\maketitle
\begin{abstract}
Many statistical problems can be addressed by applying a multiple testing procedure (MTP) that controls either the Family-wise Error Rate (FWER) or False Discovery Rate (FDR) under unknown arbitrarily-interdependent $p$-values, without explicitly modeling these inter-correlations. They include the FWER-controlling \citet{Bonferroni36} MTP and \citet{Holm79} MTP; the FDR-controlling \citet{BenjaminiYekutieli01} MTP; and the DP-MTP \citep{Karabatsos25}, based on a Dirichlet process (DP) prior distribution supporting the entire space of MTPs that control either the FWER or FDR. For such an MTP, this study introduces a new and congenial method for Bayesian predictive power analysis, for power calculation and sample size determination for any given planned future (e.g., replication or interim) study. This novel MTP predictive power analysis method is based on a joint prior distribution defining a scale matrix mixture of asymmetric multivariate normal mean-variance mixture distributions, factorized as a general prior distribution for effect sizes (e.g., obtained from expert judgment or results of prior studies), and a uniform prior distribution for correlation matrices representing arbitrary dependencies between $p$-values of test statistics of given multiple hypothesis tests under their alternative hypotheses. The new MTP power analysis method also results in $p$-value weights which can be used to minimize the relative impacts of and assess for significance-chasing biases (e.g., publication bias, $p$-hacking, etc.) in multiple testing, without needing to assume that $p$-values (effect sizes) are independent. Previous MTP power analysis methods are conditional in that they assume fixed effect sizes and a fixed correlation matrix for test statistics ($p$-values) which may be difficult to specify; while, incongenially, not fully accounting for their uncertainty, especially for MTPs that control either the FWER or FDR for arbitrarily correlated $p$-values. The new simulation-based MTP predictive power analysis method is illustrated through the analysis of $p$-values obtained by a famous study of lead exposure and re-analyzed by the previous MTP literature, using \texttt{R} package \texttt{bnpMTP}.% 332 words (Notepad++ menu: View > Summary)
\end{abstract}

\newpage
\section{Introduction}\label{Section:Introduction}

Multiple hypothesis testing problems continue to arise in many fields of the natural (life, physical), social (e.g., psychology, education), and the formal sciences, especially the bio- and statistical sciences where various hypothesis testing and multiple testing procedures (MTPs) have been developed over decades.

A hypothesis testing procedure is based on a test statistic, typically defined as a ratio of an effect size to its standard error \citep[][for example]{Cohen69.88,Hedges26}. Examples of effect size include a raw or standardized mean difference; correlation or regression coefficient; and (log-)odds ratio; etc.; each of which can be converted from one metric to another \citep{Borenstein09}. Under the given null hypothesis ($H_0$) tested (against the corresponding alternative hypothesis $H_1$, respectively), such a test statistic follows a central-zero-mean (non-central, resp.) $t$-distribution or normal distribution (e.g., asymptotically for large samples) under the null hypothesis tested, perhaps after transforming another test statistic that follows a central (non-central, resp.) $F$, $\chi_{\nu}^2$, or other standard distribution under $H_0$ ($H_1$, resp.) \citep[][]{Hoyle73}. A $p$-value is a standardized deterministic transformation of the given test statistic onto the $[0,1]$ interval \citep[][p.17]{Dickhaus14}, such that the smaller the $p$, the more decisively the given null hypothesis $H_0$ is rejected \citep[][p.31]{Efron10LSI}. The $p$-value ($p(x)$), obtained from the test statistic computed from the given sample data set $x\equiv x_n$ of size $n$, is declared as ``significant", i.e., rejecting the given null $H_0$, if $p\leq\alpha$, based on a small specified fixed Type I error probability ($\alpha$) of incorrectly rejecting $H_0$ over imaginary datasets $x_n^*$ randomly sampled from a population data-generating process where $H_0$ is true. Meanwhile, $p$-values provide a common ``bottom line" language for the communication of statistical results that are still very often communicated by scientific journals, despite concerns \citep{WassersteinLazar16}. The $p$-value is basic to statistical inference in that it can be used to compute the upper bound of the Bayes factor (\citep[][for example]{HeldOtt18, BenjaminBerger19}) and hence can provide approximate Bayesian computation.

The conditional statistical power ($1-\bar{\beta}$) of a null hypothesis testing procedure is the probability that it will detect an effect size ($\theta$) to reject the given null $H_0$, given specified $\theta$, $n$, $\alpha$, and the Type II error probability ($\bar{\beta}$) of incorrectly not rejecting $H_0$ over imaginary datasets $x_n^*$ randomly sampled from a population data-generating process where $H_1$ is true. The conditional  power of virtually any classical $\chi^2_{\nu}$-, $F$, $z$-, or $t$-test procedure can be calculated by a closed-form formula \citep[][for example]{Cohen69.88}, which can be easily rearranged to solve for any of the four quantities (power, $n$, $\theta$, and $\alpha$) given the other three. Power analysis is used to calculate the power ($1-\bar{\beta}$) of a test procedure, or the minimum data sample size needed to achieve some minimum desired power (e.g., $0.80$) for a planned future (e.g., replication or interim) study, given anticipated or estimated $\theta$ and $\alpha$; and less often, to evaluate the sensitivity of studies, or to make decisions about criteria of statistical significance \citep[][]{Murphy25}. However, conditional power analysis assumes a specified fixed effect size ($\theta$) without accounting for its uncertainty. This issue can be addressed {theoretically} by predictive power analysis \citep[][\S6.5]{SpiegelhalterFreedman86,SpiegelhalterEtAl86,SpiegelhalterEtAl04}, which averages the conditional power over a specified prior distribution representing available knowledge for the effect size {\citep[][for example]{MicheloudHeld22,PawelEtAl23,PawelEtAl24,DemartinoEtAl26}.}

Widely available software packages now easily enable common statistical analyses of data to routinely output results of typically-correlated $p$-values from multiple hypothesis tests (often, on many variables), either as the main data analysis goal, and/or as automatic by-products of other main statistical modeling objectives. Also, an MTP for marginal $p$-values can reduce the results of different (e.g., $t$, $F$, $\chi^2$, $z$, Wilcoxon, and/or log-rank, etc.) test statistics to a common interpretable $p$-value scale, while conveniently, not requiring the statistician to assume nor to specify and estimate an explicit model for the potentially-complex joint distributions of the test statistics, having typically unknown correlations. All things considered, it is no wonder why MTPs for marginal $p$-values are popular and relevant in applied statistics \citep[][and references therein]{TamhaneGou22}, which are easily usable and computable, while not requiring direct access to the original data, so that such an MTP can be readily used for meta-analysis. 

Given that there are very many MTPs, and that $p$-values in applied statistics are typically dependent (correlated), we henceforth focus on the subset of MTPs for marginal $p$-values, where each MTP provides either Family-Wise Error Rate (FWER) or False Discovery Rate (FDR) control under unknown arbitrary dependencies between $p$-values, even when the given observed $p$-values arise from a highly heterogeneous mix of different hypothesis testing procedures, and without requiring direct access to the original data. They include the un/weighted: \citet{Bonferroni36} MTP (B-MTP) and \cite{Holm79} MTP (H-MTP), each of which control the FWER; the \citet{BenjaminiYekutieli01} MTP (BY-MTP), which controls the FDR; and the Dirichlet process MTP (DP-MTP) \citep{Karabatsos25}, {defined by a DP prior distribution \citep{Ferguson73} which supports the entire space of MTPs controlling either the FWER or the FDR under arbitrary dependencies between $p$-values}. A weighted version of the B-MTP, H-MTP, BY-MTP, or DP-MTP is defined by importance weights assigned to respective $p$-values and null hypotheses \citep{BenjaminiHochberg97,KangEtAl09,TamhaneGou22,Karabatsos25}. {All these MTPs are reviewed in \S\ref{Section:ReviewMTP}.}

MTP power analysis methods have a shorter history and focus on conditional instead of predictive power. For any MTP for testing a given set of $m\geq1$ hypothesis with FWER or FDR control, power can be defined by any one of the following four ways, with respect to the $m$-variate distribution of their test statistics under the alternative hypotheses (null hypothesis, resp.) for all $m$ tests (resp.) with $m\times m$ correlation matrix. The \textit{marginal (individual) power}, of each test of hypotheses $H_{0,j}$ (for $j=1,\ldots,m$), is the marginal probability that the corresponding $p$-value leads to a rejection under level $\alpha$ by the MTP, with respect to the marginal distribution of the corresponding $j$th test statistic under the alternative hypotheses \citep{SennBretz07,Porter18}. The \textit{average power} is the un/weighted average over these $m$ marginal probabilities (resp.) under the alternative hypotheses (weights sum to $m$ over all $m$ hypotheses tested) \citep{BretzEtAl11,DudoitEtAl03,BenjaminiHochberg97,WestfallKrishen01}; which relates to the criterion of \textit{expected number of rejected hypotheses} \citep{Spjotvoll72}. The \textit{disjunctive ($d(=r)=1$/minimal/any-pair) power} is the probability that at least one false null hypotheses is rejected at level $\alpha$ by the given MTP \citep{BretzEtAl11,ChenEtAl11,MaurerMellein88,WestfallEtAl11,Ramsey78}. Finally, the \textit{conjunctive (complete/all-pairs) power}  is the probability that all false null hypotheses are rejected at level $\alpha$ by the MTP \citep{BretzEtAl11,Porter18,Ramsey78}.

For modern applied statistics which routinely outputs multiple $p$-values, MTP power analysis is more relevant than power analysis of a single test procedure. A power analysis can be used to calculate the given MTP's power, or the data sample size required to achieve some minimum desired power for the MTP for a planned future (e.g., replication or interim) study, and for the other mentioned uses of power analysis for a single test. For some FWER-controlling MTPs applied to up to three hypotheses, explicit formulas are available to compute (conditional) conjunctive power and disjunctive power  \citep{ChenEtAl11}. Simulation methods can be used to calculate the power(s) for any FWER- or FDR-controlling MTP \citep[][for example]{Porter18}. 

For any MTP that controls either the FWER or FDR for $m$ given hypothesis tests, optimal weights for respective $p$-values (hypothesis tests) can be determined through maximization of a defined power function(s) of interest {(e.g., disjunctive, conjunctive , or marginal power, etc.)} based on the specified marginal power for each test \citep[][and references therein]{RoederWasserman09,XiChen24}. In principle, an MTP that accounts for arbitrary dependence between $p$-values, while incorporating optimized weights for the respective $p$-values, provides an MTP analysis that relatively down-weights the impacts of any present significance-chasing biases (due to publication bias and/or $p$-hacking, etc.) while accounting for any heterogeneous effects, and without needing to assume independence among $p$-values (effect sizes) as done in traditional meta-analysis \citep{CooperHedgesValentine19}. This is because {the significance-chasing bias level} of any null hypothesis test's {binary-valued result (1 = significant, $p\leq\alpha$; or 0 = non-significant, $p>\alpha$) can be measured by the binary result's distance from the} test's marginal power probability, given the estimated or anticipated effect size of the test, as done by the $\chi^2$ test of excess significance \citep{IoannidisTrikalinos07}, while other tests and models of publication bias can have limited power and be difficult to use \citep{IoannidisTrikalinos07,Marks-AnglinChen20,ElliottEtAl26}. {As an aside, while the test of excess significance is based on a $\chi^2$ type distance, in principle any other distance measure from the general class of power divergence statistics \citep{CressieRead84} can be used, such as the Hellinger distance. While the $\chi^2$ test of excess significance aims to detect the presence of significance chasing biases instead of correcting for them, one key insight of this established bias testing procedure is that the marginal power of each of the $m$ hypothesis tests of a given multiple hypothesis testing application provide standards by which to address significance chasing biases. To elaborate, a significant result of a null hypothesis test with low marginal power is more likely to be a result of significance chasing biases, compared to a significant result of a null hypothesis test with high marginal power which is more likely to be a result of a true effect or signal and less likely to be a result of bias. It is less likely  because such a high powered test already makes it easier for the data analyst to detect a true effect or signal and thus lowers the need for the analyst to engage in significance chasing bias behavior. Therefore, assigning weights to $p$-values as increasing functions of their respective marginal powers would produce results of multiple hypothesis testing that would place greater relative weight to $p$-values that are less susceptible to significance chasing biases.}

For a given weighted multiple testing problem, optimal $p$-value weights can be found using any one of the available optimization algorithms that maximize any of the mentioned criteria of MTP power. {See the Appendix for a more detailed review \citep[based on][]{XiChen24}. However, while each of the available optimization methods are based on a clear objective, they are not fully satisfactory, as they are limited} to Bonferroni, graphical or chain MTPs; while numerically optimizing disjunctive or conjunctive power can be computationally costly when the number of hypothesis tests ($m$) is large; can yield non-unique $p$-value weights under equal (specified) marginal powers and correlated test statistics; can be difficult to determine or interpret (e.g., numerically optimizing for weighted or average power is easier, but involves introducing another set of weights for power that may be difficult to determine or interpret); and employs only conditional power, in that each method assumes that both the effect sizes and the correlation matrix of the $m$-variate distribution of the test statistics are fixed and known.

More ideally, a predictive MTP power analysis can be used, which would assign a prior distribution on both of these multivariate quantities in order to account for their uncertainty. Indeed, for specialized multiple testing problems employing $m$ hypothesis tests, the $m\times m$ correlation matrix of the test statistics can be directly derived or calculated, either from uncorrelated test statistics ($p$-values) or orthogonal designs or contrasts; permutation- or bootstrap-based estimates of test statistics; or for pairwise mean comparisons done via a Dunnett- or Tukey-Kramer type MTP {(see the Appendix for a more detailed review)}. However, more broadly speaking, in the modern data analysis era where MTP inferences need to be made from multiple $p$-values which often result from applications of highly heterogeneous combinations of different hypothesis testing procedures (e.g., not only for group mean comparisons), it can be very challenging to derive explicit equations or a sharp prior distribution about the correlations among corresponding test statistics, especially when the number ($m$) of hypothesis tests is large. Besides, as mentioned, this paper focuses on MTPs for marginal $p$-values, which each accounts for arbitrary dependencies between $p$-values while not needing to model or estimate these correlations directly from data.

Further, while each of the available MTP power analysis methods provides only conditional power analysis, i.e., based on fixed effect sizes and a fixed specified correlation matrix, such a method is uncongenial with any MTP (e.g., B-MTP, H-MTP, or DP-MTP) that controls either the FWER or FDR under arbitrary dependence (correlations) between the $p$-values (test statistics). This is because such an MTP accounts for all possible correlation matrices of the test statistics, not only one. And while it is tempting to view equation- and optimization-based methods for computing MTP power as more expediently convenient and precise compared to simulation methods, such an optimization method typically needs to assume a fixed correlation matrix in order to be computationally tractable. A simulation-based MTP power analysis can flexibly model more complex random correlation structures beyond a fixed correlation matrix.

Based on all the above considerations and the current state of the related MTP power analysis literature, we propose the first MTP predictive power analysis method, a straightforward simulation-based method of predictive power analysis for any MTP that provides FWER or FDR control under arbitrary-dependent $p$-values. This new MTP power analysis method, for any $m$ given hypothesis tests, is based on a joint prior distribution for the parameters of a newly developed $m$-variate (upper) non-central $t$ distribution, defined by a $m\times1$ location vector, $m\times m$ scale matrix (mainly determined by an underlying correlation matrix), and a vector of $m$ degrees of freedom parameters, for the distribution of the $m\ge1$ test statistics under their alternative hypotheses. The first part of this joint prior is defined by a prior distribution for the location vector representing the respective ratios of $m$ effect sizes to their respective standard errors, which can represent either: prior subjective clinical judgment information; be centered on the estimated effect sizes from a previous study; or more generally by a power prior, defined as the posterior distribution of the effect sizes from prior historical data \citep{IbrahimEtAl15} for, e.g., a future replication (or interim) study. The second part of the joint prior is defined by a uniform prior distribution for the underlying $m\times m$ correlation matrix for the $m$ test statistics under their respective alternative hypotheses, which can fully account for uncertainty and typical lack of prior information about this matrix, and ensure a congenial power analysis of the given MTP that controls FWER or FDR under arbitrarily dependent $p$-values (test statistics). Therefore, the new MTP power analysis method is based on a scale matrix mixture of asymmetric multivariate normal mean-variance mixtures of the test statistics under the alternative hypotheses of all the given multiple hypothesis tests (resp.). Indeed, this is the first paper that provides a method of MTP power analysis for arbitrary-correlated $p$-values, which can rapidly compute power results using an efficient algorithm \citep{PourahmadiWang15} for sampling the uniform distribution of correlation matrices.

{In Bayesian statistical inference, the prior distribution must reflect the current state of knowledge and prior information about model parameters available to the data analyst. When there is a lack of prior information about model parameters, then Bayesian theory implies that a coherent choice of prior distribution is an ``objective" non-informative prior distribution for the correlation matrix. Meanwhile, most priors used in statistical applications are objective instead of subjective, in part, because subjective elicitation is too difficult to be done even in a limited way for more than a few unknown parameters in the given statistical problem, so that such unknown parameters must be handled via objective Bayesian methods \citep[][p.216]{BergerEtAl24,BergerWolpert04}. Correspondingly, for the current context of predictive MTP power analysis, the ``objective" uniform prior distribution on the correlation matrices often honestly reflects the statistician's available prior knowledge in many MTP statistical situations. This is because of the fact that in such common situations, the MTP statistician only has access to the $p$-values of the $m$ hypothesis tests, without any access to the raw data used to compute these $p$-values, as in traditional meta-analysis. This thus leaves the statistician without the ability to estimate this dependence structure and to elicit a more informative prior distribution, thus rendering the ''objective" uniform prior over all valid correlation matrices as a reasonable choice prior in such a situation because it would honestly represent the available prior information. Indeed, such common MTP data analysis scenarios are in large part what make MTPs that control either the FWER or FDR under arbitrarily dependent $p$-values quite relevant and useful \citep[][for example]{BenjaminiYekutieli01}, including the B-MTP, H-MTP, BY-MTP, and DP-MTP. While it is arguable that the uniform prior may not meaningfully capture the more likely dependence structures among test statistics, such as equi-correlated or AR(1) correlation structures, more informative priors can be difficult to determine in more common MTP scenarios. The uniform prior is not “wrong” because it assigns non-zero probability to all arbitrary valid correlations structures for test statistics supported by the B-MTP, H-MTP, BY-MTP, and DP-MTP.}

This new MTP predictive power analysis method: (1) determines estimates of the marginal powers of the $m$ given hypothesis tests (resp.); which then can either be used for: (2) sample size determination for each of these tests given desired minimum marginal powers; (3)-(4) the other two mentioned purposes of power analysis of a single testing procedure; (5) constructing normalized weights of the corresponding $m$ $p$-values (resp.) while optimizing the power of the MTP, and down-weighting the impact of significance-chasing biases in the given MTP analysis; and/or (6) compare with observed $p$-values to provide an assessment of such biases for each individual $p$-value using Hellinger distance, in a similar spirit of the original seminal $\chi^2$ test of excess significance {(but here, not using a measure of bias over all $m$ hypothesis tests or ''studies" based on $\chi^2$ type distance)}. {It is well-known in the MTP field that a weighted MTP, which give relatively higher weights to the subset of tested null hypotheses that are likely to be false, tends to have higher statistical power compared to unweighted multiple testing procedures, for various correlation structures according to extensive simulation studies \citep[e.g,][]{Xie12}}. Further, the MTP predictive power analysis method achieves all these 6 objectives without needing to assume independence of the $m$ $p$-values, if the power-analyzed MTP controls for the FWER or FDR under arbitrarily dependent $p$-values.

Next, Section \S\ref{Section:ReviewMTP} reviews MTPs that each strongly controls FWER and/or FDR under arbitrarily dependent $p$-values, in order to contextualize the description of the new predictive MTP power analysis method for such an MTP in \S\ref{Section:MTPpower}. Then \S\ref{Section:Illustration} illustrates the new predictive MTP power analysis method to evaluate the power of the DP-MTP, and of the B-MTP and H-MTP for comparison purposes, through the MTP analysis of $p$-values arising from 41 hypothesis tests reported by a famous study on the neuropsychologic effects of unidentified childhood exposure to lead. \S\ref{Section:Theorem} proves a theorem on the convergence and convergence speed of the algorithm. \S\ref{Section:ConclusionsDiscussion} concludes by summarizing the main ideas of the paper.

{Central to the new predictive MTP power analysis method is Algorithm \ref{alg:MTPpower} presented in \S\ref{Section:MTPpower}, which can be used to simulate draws of samples from the prior predictive distribution of the joint prior mentioned above, in order to yield estimates of the various measures of powers of the $m$ test statistics, respectively, and the associated statistics mentioned above. This includes the corresponding $m$ marginal powers, being well-defined concepts in the MTP literature \citep{SennBretz07,Porter18} that can be used for estimating the Hellinger distance measures of significant-chasing bias, and for constructing $p$-value weights that can be used in a subsequent weighted MTP analysis. Meanwhile, the univariate marginal distributions of a prior predictive distribution are unique given a specifically defined prior and sampling distribution (likelihood), implying that the corresponding marginal powers and associated functions of them are uniquely defined, including the associated Hellinger distance bias measures and $p$-value weights, all while accounting for the uncertainty for the parameters underlying MTP power analysis by assigning them a joint prior distribution which averages conditional power over this prior. As mentioned earlier, this is unlike some of the previous approaches to MTP power analysis, all based on conditional power, which can yield non-unique $p$-value weights. While Algorithm \ref{alg:MTPpower} by default is based on the uniform prior on all valid correlation matrices for the $m$ given test statistics, the algorithm can be easily modified to implement more informative prior, such as priors supporting equi-correlated or AR(1) correlation structures.}

\section{Review of MTPs Valid Under Arbitrary Dependence}\label{Section:ReviewMTP} 

Here, a formal multiple hypothesis testing framework \citep{BlanchardRoquain08} is reviewed, useful for MTP inference \citep[][]{Karabatsos25}. Consider the probability space, $(\mathcal{X}, \mathfrak{X}, P)$, with probability function $P$ a member of a set  $\mathcal{P}$ of distributions representing a parametric or non-parametric model. {A null hypothesis, $H_0$, is a subset (submodel), $H_0\subset\mathcal{P}$, of distributions on $(\mathcal{X},\mathfrak{X})$, where $P\in H_0$ denotes that $P$ satisfies $H_0$, and $\mathcal{X}$ is the sample space of any observable dataset sampled from the distribution $P$, with $\mathfrak{X}$ the corresponding sigma algebra of events.} Any application of multiple testing, performed on a sample dataset $x\sim P\in\mathcal{P}$, aims to determine whether $P$ satisfies distinct null hypotheses, belonging to a certain set (family) $\mathcal{H}$ of candidate null hypotheses, with this set being typically countable, or even uncountable. 

For any given sampled dataset, $x\sim P\in\mathcal{P}$, a multiple testing procedure (MTP) is a decision function, $\mathcal{R}:x\in\mathcal{X}\mapsto\mathcal{R}(x)\subset\mathcal{H}$ that returns the subset of rejected null hypotheses; $\mathcal{R}^c(x)=\mathcal{H}-\mathcal{R}(x)$ is the subset of non-rejected hypotheses; and the MTP $\mathcal{R}(x)$ commits a Type I error if $H_0\in\mathcal{R}(x)\cap\mathcal{H}_{0}(P)$; and commits a Type II error if $H_0\notin\mathcal{R}(x)\cap\mathcal{H}_{1}(P)$, {where $\mathcal{H}_{1}(P)$ is the subset of the alternative hypotheses that are true under the given true data-generating distribution $P$}. In a typical multiple testing application, $\mathcal{H}$ is a finite or countable set of $m$ null hypotheses, $\mathcal{H}=\{H_{0,1},\ldots,H_{0,m}\}$, {tested respectively against given alternative hypotheses $\{H_{1,j}\}_{j=1}^m$,} with $m=|\mathcal{H}|\in\mathbb{Z^+}$ the number of candidate null hypotheses, $m_0(P) = |\mathcal{H}_{0}(P)|\leq m$ true null hypotheses, $m_1(P) = |\mathcal{H}_{1}(P)|=m-m_0\leq m$ (truly-)false null hypotheses, and $\pi_0=m_0/m$, under any $P\in\mathcal{P}$. For simplicity and with no loss of generality, it is henceforth assumed that $\mathcal{H}$ is countable, while the subsequently-described hypothesis testing related ideas can easily be extended to continuous hypothesis testing after adopting some more complex notation.

A typical MTP is a function $\mathcal{R}(\mathbf{p})$ of a family of $p$-values, $\mathbf{p}=(p_j,H_{0,j}\in\mathcal{H})$, where each $p_j$-value measures how probable are the observed data $x$, given that the null hypothesis $H_{0,j}$ is true \citep[][p.19, Definition 2.1]{Dickhaus14}. Assume that for each null hypothesis $H_{0,j}\in\mathcal{H}$ there exists a $p$-value function, $p_j:\mathcal{X}\rightarrow\lbrack0,1]$, having a marginally super-uniform probability distribution ($\mathbb{P}$) when $H_{0,j}$ is true, i.e.:\begin{equation}\label{eq:superUniform}
\mathbb{P}_{X\sim P}[p_j(X)\leq t]\leq t,\text{ for }\forall P\in
\mathcal{P},\text{ }\forall H_{0,j}\in\mathcal{H}_{0}(P),\text{ and }\forall t\in
\lbrack0,1],
\end{equation}
with strict equality ($\mathbb{P}_{X\sim P}[p_{H}(X)\leq t]=t$ for all $t\in\lbrack0,1]$) when the $p_j$-value is {calibrated}, i.e., uniform $\mathcal{U}[0,1]$ distributed under the null hypothesis $H_{0,j}$, which can be achieved (or improved on) using any of the available $p$-value calibration methods, e.g., for a test of a discrete model or a composite null hypothesis; or for model checking \citep[][and references therein]{DickhausEtAl12, Dickhaus13, Dickhaus14, Gosselin11, MoranEtAl24}. 

{Thus, the assumption (\ref{eq:superUniform}) for each $p$-value does not require a calibrated $p$-value for MTPs considered in the paper, since the $\mathcal{U}[0,1]$ distribution is a special case of the super-uniform distribution, even though a calibrated $p$-value would allow for a universal interpretation of the $p$-value, e.g., based on Fisher's scale of evidence for $p$-values \citep[][p.31]{Efron10LSI}. The $p$-values analyzed in the case study illustration in \S\ref{Section:Illustration} are calibrated asymptotically. Indeed, it is possible to define $p$-values only via the super-uniformity property (\ref{eq:superUniform}) \citep[][Definition 8.3.26]{CasellaBerger02}, which does not make specific requirements on the distribution of the underlying test statistic under the null hypothesis. Though,} the predictive MTP power analysis method introduced in  \S\ref{Section:MTPpower} requires the specification of an explicit functional relationship between each $p$-value and its test statistic, between the joint distribution of test statistics and corresponding $p$-values for any multiple testing problem applying $m$ hypothesis tests. Then, a $p$-value function can be written as $p(x)\equiv p\{t(x)\}$ from a test statistic $t\equiv t(x)$ of any random dataset $x$.

When using any MTP $\mathcal{R}$ to test a set of hypotheses, $\mathcal{H}$, a traditional criterion for Type I error control is the FWER, the probability ($\mathbb{P}$) of making at least one false discovery \citep{HochbergTamhane87}:
\begin{equation}\label{eq:FWER}
\mathrm{FWER}_P(\mathcal{R}) = \mathbb{P}_{X\sim P}[\text{Reject any true hypothesis, } H_{0,j}\in\mathcal{H}_{0}(P)\subseteq\mathcal{H}];
\end{equation}
while an alternative criterion is the FDR \citep{BenjaminiHochberg95}, the expected ($\mathbb{E}$) proportion of false rejections of null hypotheses out of the total number of rejected hypotheses:
\begin{equation}
    \mathrm{FDR}_P(\mathcal{R})=\mathbb{E}_{X\sim P}[\mathrm{FDP}_P(\mathcal{R}(X))]=\mathbb{E}_{X\sim P}\left[\frac{|\mathcal{R}(X)\cap\mathcal{H}_{0}(P)|}{\mathrm{max}\{|\mathcal{R}(X)|,1\}}\right].
\end{equation}
The practice of multiple hypothesis testing aims to maximize the expected number of rejections while controlling the FWER or FDR at a preset level $\alpha\in(0,1)$, typically $\alpha = 0.05$ or $0.01$, etc. An MTP strongly controls the FWER (\textit{FDR}, resp.) if $\mathrm{FWER}_P(\mathcal{R})\leq\alpha$ (if $\mathrm{FDR}_P(\mathcal{R})\leq\pi_0\alpha\leq\alpha$, resp.) for any $\alpha\in(0,1)$ and all $P\in\mathcal{P}$ (all $\pi_0\in[0,1]$, resp.), while $\mathrm{FDR}\leq\mathrm{FWER}$, with $\mathrm{FDR}=\mathrm{FWER}$ if all null hypotheses are true, and thus $\mathrm{FWER}\leq\alpha$ implies $\mathrm{FDR}\leq\alpha$, i.e., FDR control is more liberal \citep{TamhaneGou22}.

Many MTPs can each be characterized as a step-up MTP, $\mathcal{R}_{SU}^{\Delta_\alpha}$, which for the order statistics $p_{(1)}\leq\dots\leq p_{(m)}$ of $m$ $p$-values of the given multiple testing problem, specifies a non-decreasing sequence of thresholds, 
$0\leq\Delta_\alpha(H_{0,(1)})\leq\cdots\leq\Delta_\alpha(H_{0,(m)})\leq 1$ for the respective $m$ ordered tested null hypotheses, and then rejects the null hypotheses having the $R_\alpha(x)$ smallest $p$-values, with:
\begin{equation}\label{eq:R.Delta}
  \textit{R}_\alpha(x) = \underset{r\in\{0,1,\ldots,m\}}{\mathrm{max}}\{r : p_{(r)}(x) \leq \Delta_\alpha(H_{0,(r)})\}.
\end{equation}
where $p_{(0)}\equiv0$, and given a specified threshold function:
\begin{equation}\label{eq:Delta}
\Delta_{\alpha,\nu}(H_{0.(r)})\equiv\alpha\pi(H_{0,(r)})\beta_{\nu}(r),
\end{equation}
with weight function $\pi:\mathcal{H}\rightarrow[0,1]$ and $\beta_{\nu}:\mathbb{R}^+\rightarrow\mathbb{R}^+$ is a shape (reshaping) function. MTPs that strongly control FWER under arbitrary dependencies between $p$-values, include: the (unweighted) \citet{Bonferroni36} MTP (B-MTP), defined by thresholds $\Delta_{\alpha,\nu}(H_{0,(r)})\equiv\alpha\beta_{\nu}(r)/m=\alpha/m$ with shapes $\beta_{\nu}(r)\equiv1$ and weights $\pi(H_{0,(r)})\equiv w_{(r)}=1/m$; the \citet{Holm79} step-down MTP (H-MTP), defined by thresholds $\Delta_{\alpha,\nu}(H_{0,(r)})\equiv\frac{\alpha}{m-r+1}$. The weighted B-MTP is defined by thresholds $\Delta_{\alpha,\nu}(H_{0,(r)})=\alpha\pi(H_{0,(r)})\beta_{\nu}(r)=\alpha w_{(r)}$ using weights $\pi(H_j)\equiv w_{j}\in[0,1]$ assigned to respective null hypotheses $H_{0,j}\in\mathcal{H}$ such that $\sum_{i=1}^{m}w_j=1$, where $\pi(H_j)$ is an arbitrary probability distribution on $\mathcal{H}$ defining the relative importances of the $m$ null hypotheses tested \citep{TamhaneGou22}. The weighted H-MTP is defined by thresholds $\Delta_{\alpha,\nu}(H_{0,(r)})\equiv\frac{w_{(r)}}{\sum_{k=r}^mw_{(k)}}\alpha$ with weights $w_j\ge0$, $j=1,\ldots,m$.

For any probability measure $\nu$ on $(0,\infty)$, the step-up procedure (\ref{eq:R.Delta})-(\ref{eq:Delta}) based on the shape function:
\begin{equation}\label{eq:beta}
\beta_{\nu}(r)\equiv\int_0^{r} x \text{d}\nu(x)
\end{equation}
strongly controls $\mathrm{FDR}\leq\alpha\pi_0\leq\alpha$ under arbitrary dependencies between $p$-values, where $\pi:\mathcal{H}\rightarrow\lbrack0,1]$ is a probability mass function with respect to a counting measure $\Lambda$ on $\mathcal{H}$ and $\pi_0=\sum_{H_{0,j}\in\mathcal{H}_0}\Lambda(\{H_{0,j}\})\pi(H_{0,j})$. Using different weights $\pi(H_{0,j})$ (weights $\Lambda(\{H_{0,j}\})$, resp.) over countable $\mathcal{H}$ gives rise to \textit{weighted $p$-values} (\textit{weighted FDR}; \citet{BenjaminiHochberg97}, resp.). For any finite set of $m$ null hypotheses $\mathcal{H}$ and $p$-values, the \citet[BY;][]{BenjaminiYekutieli01} distribution-free step-up MTP (i.e., BY-MTP) is defined by probability measure $\nu(\{k\}) = (k\sum_{j=1}^m\tfrac{1}{j})^{-1}$ with support in $\{1,\ldots,m\}$ and threshold function $\Delta_{\alpha,\nu}(H_{0,(r)},r)=\alpha\beta_{\nu}(r)/m$, based on linear shape function $\beta_{\nu}(r)=\sum_{k=1}^r k\nu(\{k\})=r/ (\sum_{j=1}^m \tfrac{1}{j})$ and $p$-value weights $\pi(H_{0,j})\equiv w_{j}=1/m$ for $j=1,\ldots,m$ \citep[][Lemma 3.2, p.976]{BlanchardRoquain08}. The weighted BY-MTP is instead defined by an arbitrary probability distribution $\pi(H_{0,j})$ on $\mathcal{H}$.

Each of the above MTPs provide conservatively control the FWER or FDR, especially when the number of hypothesis tests ($m$) is large, while other choices of $\nu$ can sometimes improve the power of the corresponding MTP \citep[][\S4.2]{BlanchardRoquain08}. With this in mind, the DP-MTP method was developed  \citep{Karabatsos25}, this MTP defined by a Bayesian nonparametric, Dirichlet process (DP) \citep{Ferguson73} prior distribution that supports the entire space of random probability measures (r.p.m.s), $\nu$, which inn turn, for an observed set of $p$ values $\{p_j\}_{j=1}^m$, induces a DP prior distribution of random MTP thresholds $\Delta_{\alpha,\nu}(H_{0,(r)})\equiv\alpha\pi(H_{0,(r)})\beta_{\nu}(r)$ via (\ref{eq:R.Delta})-(\ref{eq:beta}), where each random threshold $\Delta_{\alpha,\nu}$ (and DP random $\nu$), and thus the DP-MTP, strongly controls either the FWER or FDR under arbitrary dependence among $p$-values. {Thus, the DP-MTP procedure naturally accounts for the uncertainty in the selection of MTPs and their respective decisions regarding which number of the smallest $p$-values are significant discoveries, from any set of hypotheses tested. Further, DP-MTP also measures each $p$-value's probability of significance relative to the DP prior predictive distribution of this space of all MTPs.}

Specifically, for any probability space $(\mathcal{X},\mathcal{A},G)$, an r.p.m. $\nu$ follows a Dirichlet process (DP) prior with baseline probability measure $\nu_0$ and mass parameter $M$, denoted $\nu\sim\mathrm{DP}(M\nu_0)$, if 
\begin{equation}\label{eq:DP}
(\nu(B_1),\ldots,\nu(B_m))\sim\mathrm{Dirichlet}_m(M\nu_0(B_1),\ldots,M\nu_0(B_m))
\end{equation}
for any (pairwise-disjoint) partition $B_1,\ldots,B_m$ of the sample space $\mathcal{X}$, with expectation $\mathbb{E}[\nu(\cdot)]=\nu_0(\cdot)$, variance $\mathbb{V}[\nu(\cdot)]=\frac{\nu_0(\cdot)[1-\nu_0(\cdot)}{M+1}$, and almost-sure support of the space of discrete r.p.m.s $\nu$ \citep{Ferguson73}. The DP-MTP is based on the DP prior (\ref{eq:DP}) centered on baseline measure $\nu_0(r-1,r]\equiv (r\sum_{j=1}^m\tfrac{1}{j})^{-1}$ (for $r=0,1,\ldots,m$) chosen to match in prior expectation (given $M$) the probability measure $\nu$ defining the BY MTP; and based on an exponential  hyper-prior distribution of the DP mass parameter $M$ in (\ref{eq:DP}), given by:
\begin{equation}\label{eq:DPhyperprior}
M\sim\mathrm{Exponential}(\mu),
\end{equation}
and a default hyper-prior is chosen by $\mu\equiv1$. \citet{Karabatsos25} provides further discussions of the choices and characterizations of the DP prior for DP-MTP.

With respect to the joint hierarchical DP prior distribution (\ref{eq:DP})-(\ref{eq:DPhyperprior}) with BY-MTP baseline $\nu_0$, the DP-MTP, for each $p$-value from a given set of $m$ ordered $p$-values, $\{p_{(r)}\}_{r=1}^m$, counts the proportion of times the $p$-value is significant, represented by the following vector (denoted $[\cdot]$) of prior predictive probabilities:
\begin{subequations} \label{eq:I.DP-MTP}
\begin{eqnarray}
&&[\Pr\{\mathbf{1}(r\leq\textit{R}_{\alpha,\nu}(x); p_{(r)})=1\}: \text{ }r=0,1,\ldots,m]\\
&&=\int\dots\int[\mathbf{1}(r\leq\textit{R}_{\alpha,\nu}(x); p_{(r)}):r=0,1,\ldots,m]\mathcal{D}_{m}(\mathrm{d}\nu_1,\ldots,\mathrm{d}\nu_m\mid M\nu_{0}),
\end{eqnarray}
\end{subequations}
where $\bm{1}(\cdot)$ is the (1 or 0 valued) indicator function, $\mathcal{D}_{m}(\nu_1,\ldots,\nu_m\mid M\nu_{0})$ (with $\nu_r\equiv B_j=\nu(r-1,r]$ for $r=1,\ldots,m$) denotes the cumulative distribution function (CDF) of the Dirichlet distribution (\ref{eq:DP}), and:
\begin{subequations} \label{eq:R.DP-MTP}
\begin{eqnarray}
R_{\alpha,\nu}(x) &=& \underset{r\in\{0,1,\ldots,m\}}{\mathrm{max}}\{r : p_{(r)}(x) \leq \Delta_{\alpha,\nu}(H_{0,(r)})\}\\
 &=& \underset{r\in\{0,1,\ldots,m\}}{\mathrm{max}}\left\{r : p_{(r)}(x) \leq \alpha\pi(H_{0,(r)})\beta_{\nu}(r)\right\}\\
 &=& \underset{r\in\{0,1,\ldots,m\}}{\mathrm{max}}\left\{r : p_{(r)}(x) \leq \alpha\pi(H_{0,(r)})\sum_{j=1}^{r}j\nu (j-1,j]\right\},
\end{eqnarray}
\end{subequations}
with $p_{(0)}\equiv0$, while inducing a prior predictive distribution of the number $R_{\alpha,\nu}$ in (\ref{eq:R.DP-MTP}) of the smallest $p$-values from $\{p_{(r)}\}_{r=1}^m$ that represent significant discoveries. Also, (\ref{eq:R.DP-MTP}) is defined by weights $\pi(H_{0,(r)})$ and corresponding test levels $\alpha_{(r)}=\alpha\pi(H_{0,(r)})$ given the overall prespecified level $\alpha$ (e.g., $\alpha=0.05$ or $0.01$, etc.) for the respective $m$ ordered $p$-values $\{p_{(r)}\}_{r=1}^m$ satisfying $\sum_{r=1}^m\pi(H_{0,(r)})\equiv1$, such that for $r=1,\ldots,m$, equal weights $\pi(H_{0,(r)})=1/m$ defines the unweighted DP-MTP \citep{Karabatsos25}, while a weighted DP-MTP is defined by an arbitrary probability distribution on $\mathcal{H}$ (as in the weighted B-MTP).

The joint hierarchical DP prior distribution (\ref{eq:DP})-(\ref{eq:DPhyperprior}) of the DP-MTP induces a prior predictive distribution for the shape parameter $\beta_{\nu}$ and corresponding threshold parameter $\Delta_{\alpha,\nu}$ and number of discoveries $R_{\alpha}$ from (\ref{eq:R.Delta}), thereby treating these functions as random instead of fixed as done by the standard MTPs, while accounting for uncertainty in the selection of MTPs and their respective cut-off points and decisions regarding which of the smallest $p$-values are significant discoveries from a given set $\mathcal{H}$ of null hypotheses tested. The DP-MTP method thus emphasizes prior predictive hypothesis testing \citep{Box80} and multiple bias modeling \citep{Greenland05} while extending Vibrations of Effects analysis \citep{VinatierEtAl25} to provide multiple hypothesis testing with uncertainty quantification. See \citet{Karabatsos25} for further discussion. 

The DP-MTP method can be run by \texttt{R} package \texttt{bnpMTP} \citep{Karabatsos25bnpMTP} using the code line \texttt{bnpMTP($\cdot$)}, which, given $m$ input $p$-values (and perhaps optional inputs of $\alpha$, $p$-value weights, $N$, and $\mu$), uses standard methods to generate $N$ Monte Carlo samples of r.p.m.s $\{\nu_g\}_{g=1}^m$ from the joint prior distribution (\ref{eq:DP})-(\ref{eq:DPhyperprior}) (with BY-MTP $\nu_0$) and of  corresponding $N$ samples from the prior predictive distribution of (\ref{eq:I.DP-MTP})-(\ref{eq:R.DP-MTP}). In principle, DP-MTP can be extended to online multiple hypothesis testing \citep[][and references therein]{RobertsonEtAl23} of a potentially-infinite number of hypothesis tests, by relaxing the constraints of the weights $\pi(H_{0,(r)})$ of the respective $m$ ordered $p$ values $\{p_{(r)}\}_{r=1}^{m(t)}$ obtained at at any given time point $t$ to satisfy $\sum_{r=1}^m\pi(H_{0,(r)})<1$ (with total level $\sum_{r=1}^m\alpha\pi(H_{0,(r)})<\alpha$) before all tests are performed, and to satisfy $\sum_{r=1}^m\pi(H_{0,(r)})\equiv1$ (with intended total level $\sum_{r=1}^m\alpha\pi(H_{0,(r)})=\alpha$) after all tests are done (the same method can be used to define an online version of any other MTP mentioned in this section). This defines the only online MTP that controls FWER and/or FDR under arbitrarily-dependent $p$-values with uncertainty quantification in multiple testing. Further, it is straightforward to extend the DP-MTP method to handle tests of continuous hypotheses, either by modeling $\nu$ as a discrete r.p.m. (e.g., assigned a DP prior), or by assigning $\nu$ a more general Bayesian nonparametric prior distribution which supports the space of continuous random probability measures, such as a DP mixture of continuous densities \citep{Lo84}.

\section{Predictive MTP power analysis Under Arbitrary Dependence}\label{Section:MTPpower} 

Consider any (off/online, un/weighted, and discrete or continuous) multiple hypothesis testing scenario employing $m\geq1$ hypothesis tests using a preset total Type I error rate control $\alpha$, and corresponding random vector of $m\geq1$ test statistics $\bm{T}=(T_1,\ldots,T_m)^\top$ with possible realizations $\bm{T}=\bm{t}=(t_1,\ldots,t_m)^\top\in\mathbb{R}^m$. Each test statistic $T_j\in\bm{T}$ typically has (or can be specified to have) the general form $T_j\equiv\theta_j/\sigma_{\theta_j}$, with $\theta$ an effect size parameter of interest (e.g., \S1). Any realized value $t_j\equiv t_{n_j}=\widehat{\theta}_j/\sigma_{\widehat{\theta}_j}$ of $T_j$ is determined by some estimate $\widehat{\theta}_j$ of $\theta$ with (estimated) standard error $\widehat{\sigma}_{\widehat{\theta}_j}$ based on a dataset of size $n_j$. 

Further, assume that each of the $m$ test procedures tests a null hypothesis $H_{0,j}\in\mathcal{H}$ against an alternative hypothesis $H_{1,j}$ using a test statistic $T_j\in\bm{T}$ which under $H_{0,j}$ follows a standard central location($\mu$)-scale($\sigma$) Student $t$-distribution, $T_j\sim\mathcal{T}(\mu\equiv 0,\sigma\equiv 1,\nu_j)$ with degrees of freedom $\nu_{j}>0$ an increasing function of $n_j$; and under $H_{1,j}$ follows a non-central $t$-distribution, $T_j\sim\mathcal{NT}(\mu_j,\sigma\equiv 1,\nu_j)$, with non-zero non-centrality parameter $\mu_j\in\mathbb{R}/\{0\}$; perhaps after transforming the original test statistic to have a normal-like distribution under $H_{0,j}$ and under $H_{1j}$. The non-central $\mathcal{NT}(\mu_j,1,\nu)$ distribution (and asymptotic normal distribution $\mathcal{N}(\mu_j,1)$ distribution, resp.) is thus the basis for power analysis for a $t$-test ($z$-test, resp.) of a mean difference, correlation or regression coefficient, (log-)odds ratio; etc., as mentioned in \S1 and by statistics textbooks. Further examples using $z$-tests and $t$-tests are illustrated in \S\ref{Section:Illustration}. 

Recall that if $\mu\neq0$, then the $\mathcal{NT}(\mu,1,\nu)$ distribution is unimodal and asymmetric; {and if $\mu=0$,} this distribution is symmetric and coincides with the $\mathcal{T}{}(0,1,\nu)$ distribution. For each test procedure $j=1,\ldots,m$, we have that: $n_j\rightarrow\infty$ implies $\nu_{j}\overset{d}{\rightarrow}\infty$ and asymptotic convergences to normal distributions, $\mathcal{T}(0,1,\nu)\overset{d}{\rightarrow}\mathcal{N}(0,1)$ under $H_{0,j}$, and $\mathcal{NT}(\mu_j,1,\nu)\overset{d}{\rightarrow}\mathcal{N}(\mu_j,1)$ under $H_{1,j}$, while the test procedure may alternatively employ these asymptotic normal distributions; and the two-tailed $p$-value is given by $p_j=2F_{\nu_j}(-|t_j|)$ with asymptotic $p$-value $p_j=2\Phi(-|t_j|)$); and alternatively, a one-sided lower-tailed (upper-tailed, resp.) $p$-value is given by $p_j=F_{\nu_j}(t)$ with asymptotic $p$-value $p_j=\Phi(t_j)$ (or by $p_j=1-F_{\nu_j}(t_j)$ with asymptotic $p$-value $p_j=1-\Phi(t_j)$, resp.), where $F_{\nu}$ is the cumulative distribution function (CDF) of the $\mathcal{T}(0,1,\nu)$ distribution and $\Phi$ is the normal $\mathcal{N}(0,1)$ distribution CDF.

Consider $m(\geq1)$ variate generalizations of the normal, $t$, and non-central $t$ distributions. Let $\mathcal{N}_m(\bm{\mu},\bm{\Sigma})$ be the $m$-variate normal distribution of a random vector $\bm{Y}=(Y_1,\ldots,Y_m)^\top$, defined by parameters of mean vector $\bm{\mu}=(\mu_1,\ldots,\mu_m)^\top$ and $m\times m$ symmetric positive-definite (s.p.d.) covariance matrix, $\bm{\Sigma}=(\sigma_{j,k})_{m\times m}$, with diagonal elements $(\sigma^2_{1,1}\equiv\sigma^2_1,\ldots,\sigma^2_{m,m}\equiv\sigma^2_m)^\top$ the respective variances of $(T_1,\ldots,T_m)^\top$, and each off-diagonal element $\sigma_{j,k}$ is the covariance $\mathbb{C}(T_j,T_k)$ for each distinct $j,k\in\{1,\ldots,m\}$. It can be convenient to model the (co)variance matrix $\bm{\Sigma}$ through either of the three following separation strategies: $\bm{\Sigma}\equiv\mathrm{diag}(\mathrm{\bm{\sigma}})\bm{P}\mathrm{diag}(\mathrm{\bm{\sigma}})$, $\bm{\Sigma}\equiv\sigma^2\bm{P}$, or $\bm{\Sigma}\equiv\bm{P}$; where $\mathrm{diag}(\bm{\sigma})$ is the $m\times m$ diagonal matrix of standard deviations $\mathrm{\bm{\sigma}}=(\sigma_{1,1},\ldots,\sigma_{m,m})^\top$, $\sigma^2$ is a common variance, and $\bm{P}=(\rho_{j,k})_{m\times m}$ is the $m\times m$ correlation matrix (i.e., a s.p.d. matrix with $\rho_{j,j}=1$ and $\rho_{j,k}\in[-1,1]$) for all distinct $j,k\in\{1,\ldots,m\}$). A sample random vector $\bm{Y}\sim\mathcal{N}_m(\bm{\mu},\bm{\Sigma})$ can be generated by $\bm{Y}=\bm{\mu} + A\bm{Z}$, with $\bm{Z}=(Z_1,\ldots,Z_m)^\top$ and $\{Z_j\}_{j=1}^m\overset{\text{iid}}{\sim}\mathcal{N}(0,1)$, using the Cholesky decomposition $\bm{\Sigma}=AA^\top$ of the s.p.d. matrix $\bm{\Sigma}$, where $A=(A_{j,k})_{m\times m}$ is the lower-triangular matrix with $A_{jj}>0$ for $j=1,\ldots,m$.

Let $\mathcal{T}_m(\bm{\mu},\bm{\Sigma},\bm{\nu})$ be the $m$-variate Student's $t$-distribution of random $\bm{T}=(T_1,\ldots,T_m)^\top$, a symmetric distribution defined by parameters of location (mean) vector $\bm{\mu}=(\mu_1,\ldots,\mu_m)^\top$, scale matrix $\bm{\Sigma}=(\sigma_{j,k})_{m\times m}$, and degrees of freedoms $\bm{\nu}=(\nu_1,\ldots,\nu_m)^\top$; with corresponding pairwise covariances $\mathbb{C}(T_j,T_k)=\tfrac{\sigma_{j,k}\sqrt{\nu_j\nu_k}}{2}\tfrac{\Gamma\{(\nu_j-1)/2\}}{\Gamma(\nu_j/2)}\tfrac{\Gamma\{(\nu_k-1)/2\}}{\Gamma(\nu_k/2)}$ (if $\nu_j>1,\nu_k>1$) and correlations $\mathbb{P}(T_j,T_k)=\tfrac{\rho_{j,k}}{2}\tfrac{\Gamma\{(\nu_j-1)/2\}\sqrt{\nu_j-2}}{\Gamma(\nu_j/2)}\tfrac{\Gamma\{(\nu_k-1)/2\}\sqrt{\nu_k-2}}{\Gamma(\nu_k/2)}$ (if $\nu_j>2,\nu_k>2$; and $\rho_{j,k}\equiv\tfrac{\sigma_{j,k}}{\sqrt{\sigma_{j.j}\sigma_{k,k}}}$) for all distinct $j,k\in\{1,\ldots,m\}$, where $\Gamma(\cdot)$ is the gamma function. If $\bm{\Sigma}$ is a diagonal matrix, then the univariate marginal distributions of $\mathcal{T}_m(\bm{\mu},\bm{\Sigma},\bm{\nu})$ are independent symmetric Student $t$-distributions, $\mathcal{T}(\mu_j,\sigma_{j,j},\nu_j)$ for $j=1,\ldots,m$ \citep[][Corollary 3]{Ogasawara24}. A sample random vector $\bm{T}\sim\mathcal{T}_m(\bm{\mu},\bm{\Sigma},\bm{\nu})$ can be generated by:
\begin{equation}\label{eq:MTdf}
\bm{T}=\bm{\mu} + A\bm{Z}\oslash(\sqrt{V_1/\nu_1},\ldots,\sqrt{V_m/\nu_m})^\top,\end{equation}
with $\bm{Z}=(Z_1,\ldots,Z_m)^\top$, $Z_j\sim\mathcal{N}(0,1)$ and $V_j\sim\chi^2_{\nu_j}$ for $j=1,\ldots,m$ \citep[][Def.1]{Ogasawara24}, where $\oslash$ denotes the Hadamard (element-wise) division of vectors. A special case of $\mathcal{T}_m(\bm{\mu},\bm{\Sigma},\bm{\nu})$ is the $\mathcal{T}_m(\bm{\mu},\bm{\Sigma},\nu)$ distribution \citep[][eq.1.1]{Cornish54, KotzNadarajah04} defined by stochastic representation $\bm{T}=\bm{\mu} + A\bm{Z}/\sqrt{V/\nu}$, $\bm{Z}=(Z_1,\ldots,Z_m)^\top$, 
$\{Z_j\}_{j=1}^m\overset{\text{iid}}{\sim}\mathcal{N}(0,1)$, $V\sim\chi^2_{\nu}$, $\nu>0$, and non-independent, $m$ univariate marginal $\mathcal{T}(\mu_j,\sigma_{j},\nu)$ distributions, $j=1,\ldots,m$.

Now, introduce the $m$-variate upper non-central $t$-distribution, $\bm{T}\sim\mathcal{NT}_m(\bm{\mu},\bm{\Sigma},\bm{\nu})$, defined by parameters of non-centrality vector $\bm{\mu}\in\mathbb{R}^m$, scale matrix $\bm{\Sigma}$, and degrees of freedoms, $\bm{\nu}=(\nu_1,\ldots,\nu_m)^\top$. This $m$-variate distribution is a member of the class of normal mean-variance mixtures \citep[][\S3.2.2]{McNeilEtAl05}, and defined by the stochastic representation:
\begin{equation}\label{eq:MNTdf}
\bm{T} = (\bm{\mu} + A\bm{Z}) \oslash(\sqrt{V_1/\nu_1},\ldots,\sqrt{V_m/\nu_m})^\top,
\end{equation}
with $\bm{Z}=(Z_1,\ldots,Z_m)^\top$, $Z_j\sim\mathcal{N}(0,1)$ and $V_j\sim\chi^2_{\nu_j}$ for $j=1,\ldots,m$. The $\mathcal{NT}_m(\bm{\mu}$, $\bm{\Sigma},\bm{\nu})$ distribution is asymmetric when $\bm{\mu}\neq\bm{0}$, and coincides with the symmetric $\mathcal{T}_m(\bm{\mu},\bm{\Sigma},\bm{\nu})$ distribution when $\bm{\mu}=\bm{0}$. A special case of $\mathcal{NT}_m(\bm{\mu},\bm{\Sigma},\bm{\nu})$ is the $m$-variate upper non-central $t$-distribution, $\mathcal{NT}_m(\bm{\mu},\bm{\Sigma},\nu)$, defined by a scalar degrees of freedom $\bm{\nu}\equiv\nu$ and univariate marginal non-central $t$-distributions $\mathcal{NT}(\mu_j,\nu)$ for $j=1,\ldots,m$ \citep[][p.81; \S3.1; eq.5.1; resp.]{Kshirsagar61,Juritz73, KotzNadarajah04}, and by a complicated $m$-variate probability density function (pdf) even if $\Sigma\equiv\sigma^2\bm{P}$. Therefore, most applications of this distribution $\bm{T}\sim\mathcal{NT}_m(\bm{\mu},\bm{\Sigma},\nu)$ employ its stochastic representation, given by $\bm{T} = (\tfrac{V}{\nu})^{-1/2}(\bm{\mu} + A\bm{Z})$, where $\bm{Z}=(Z_1,\ldots,Z_m)^\top$ with $\{Z_j\}_{j=1}^m\overset{\text{iid}}{\sim}\mathcal{N}(0,1)$ and $V\sim\chi^2_{\nu}$ \citep[][\S3.1; p.133, eq.8; resp.]{Juritz73,Hofert13}, e.g., this distribution can be sampled using the \texttt{mvtnorm} \texttt{R} package code line \texttt{rmvt(..., type = "Kshirsagar")} \citep{GenzBretz09}. A straightforward modification of this code to divide a $m$-variate normal random $(\bm{\mu} + A\bm{Z})$ by a random $(\sqrt{V_1/\nu_1},\ldots,\sqrt{V_m/\nu_m})^\top$ instead of $V/\nu$, enables sampling from the $\mathcal{NT}_m(\bm{\mu},\bm{\Sigma},\bm{\nu})$ distribution with vector $\bm{\nu}$  via its stochastic representation (\ref{eq:MNTdf}).  

Let $\widehat{\bm{\theta}}=(\widehat{\theta}_1,\ldots,\widehat{\theta}_m)^\top$ be the vector of anticipated or estimated effect sizes for a planned future study that will conduct tests of a specific set of $m$ null hypotheses. In particular, $\widehat{\bm{\theta}}$ may represent the data analyst's anticipated effect sizes, or may be the estimates of effect sizes obtained from a prior study (used for a future replication or interim study) of the same $m$ hypothesis tests. Also, let $\bm{\sigma}_{\widehat{\bm{\theta}}}=(\sigma_{\widehat{\theta}_1},\ldots,\sigma_{\widehat{\theta}_m})^\top$ be the standard errors of these $m$ effect sizes $\widehat{\bm{\theta}}$, which are {respectively decreasing functions} of the (anticipated or actual) sample sizes $(n_1,\ldots,n_m)^\top$, and are the square roots of the diagonal elements of the (e.g., inverse Fisher information) variance-covariance matrix $\Sigma_{\widehat{\theta}}$ of $\widehat{\bm{\theta}}$. In more common applications of power analysis, $\bm{\sigma}_{\widehat{\bm{\theta}}}$ is known instead of the full matrix $\Sigma_{\widehat{\theta}}$, though, covariances can sometimes be derived for certain effect size measures \citep{GleserOlkin09}. Specifying such covariances may not be crucial as this paper emphasizes power analysis of MTPs controlling FWER or FDR under arbitrarily dependent $p$-values (test statistics and effect sizes).

A conditional MTP power analysis $m$ null hypotheses tests is based on test statistics distributed as $\bm{T}\sim\mathcal{T}_m(\bm{0},\bm{P},\bm{\nu})$ under all null hypotheses $\{H_{0,j}\}_{j=1}^m$ jointly, and distributed as $\bm{T}\sim\mathcal{NT}_m(\bm{\mu},\bm{P},\bm{\nu})$ under all alternative hypotheses $\{H_{1,j}\}_{j=1}^m$ jointly, given a specified fixed: location vector $\bm{\mu}=|\widehat{\bm{\theta}}\oslash\bm{\sigma}_{\widehat{\bm{\theta}}}|$; $m\times m$ correlation matrix $\bm{P}$ of the underlying $m$-variate normal $\mathcal{N}_m(\bm{\mu},\bm{P})$ distribution; and degrees of parameters $\bm{\nu}=(\nu_1,\ldots,\nu_m)^\top$ (where $\nu_j\rightarrow\infty$ based on sample size $n_j\rightarrow\infty$ for any test based on a normal distribution approximation of its test statistic $T_j$ under the null $H_{0,j}$). 
% (\citet{Porter18}, pp.277-278, who instead used a symmetric $\mathcal{T}_m(\bm{\mu},\bm{P},\nu)$ under $\{H_{1,j}\}_{j=1}^m$).
For $m=1$ hypothesis test, $\mathcal{NT}_m(\bm{\mu},\bm{P},\bm{\nu})$ becomes the univariate non-central $\mathcal{NT}(\mu,1,\nu)$ distribution ($\mathcal{N}(\mu,1)$ distribution, resp.) used by classical conditional power analysis of a $t$-test ($z$-test if $\nu\rightarrow\infty$, resp.) \citep{Cohen69.88}. However, conditional MTP power analysis relies on a prespecified fixed parameters $(\bm{\mu},\bm{P},\bm{\nu})$ which can be difficult to specify when the number of tests $m$ is sufficiently large, and does not account for their uncertainty; and cannot provide a congenial power analysis of an MTP that controls the FWER or FDR under arbitrary dependence between $p$-values, because such an MTP accounts for all possible $m\times m$ s.p.d. correlation matrices of test statistics underlying these $m$ $p$-values, instead of only one fixed correlation matrix.

These issues can be addressed by an predictive power analysis of such an MTP, based on a joint prior distribution for the parameters $(\bm{\mu},\bm{P})$ driving the $m$-variate distribution of test statistics, $\bm{T}\sim\mathcal{NT}_m(\bm{\mu},\bm{P},\bm{\nu})$, and corresponding $p$ values under all alternative hypotheses $\{H_{1,j}\}_{j=1}^m$. Such a prior accounts for uncertainty in these parameters, and leads to the calculation of marginal powers of $m$ hypothesis tests, thereby addressing all six aims of MTP power analysis (mentioned at the end of \S1) without needing to assume independent $p$-values. The joint prior distribution (probability density function) for $(\bm{\mu},\bm{P})$ is specified by:
\begin{equation}\label{eq:prior}
\pi(\bm{\mu},\bm{P})=\mathcal{N}_m(\bm{\mu}\mid\widehat{\bm{\theta}}\oslash\bm{\sigma}_{\widehat{\bm{\theta}}},\bm{I}_m)\mathrm{LKJ}(\bm{P}\mid\eta\equiv1),
\end{equation}
based on an $m$-variate normal distribution for the mean parameters $\bm{\mu}$, with prior mean vector defined by anticipated or estimated effect sizes and their respective standard errors, and with (co)variances given by the $m\times m$ identity matrix ($\bm{I}_m$); and a uniform prior distribution for the $m\times m$ (s.p.d.) correlation matrix parameters $\bm{P}$ (i.e., the LKJ distribution \citep{LewandowskiEtAl09} with parameter $\eta\equiv 1$), which ensures a congenial (predictive) power analysis for such an MTP that controls the FWER or FDR under arbitrary dependent $p$-values, because this uniform prior supports all possible $m\times m$ correlation matrices of the test statistics underlying these $p$-values. In contrast, a conditional MTP power analysis assigns a point-mass prior at a chosen point $(\widehat{\bm{\mu}}=\widehat{\bm{\theta}}\oslash\bm{\sigma}_{\widehat{\bm{\theta}}},\bm{P})$ with a fixed correlation matrix $\bm{P}$.

The default, $m$-variate normal prior $\mathcal{N}_m(\bm{\mu}\mid\widehat{\bm{\theta}}\oslash\bm{\sigma}_{\widehat{\bm{\theta}}},\bm{I}_m)$ in (\ref{eq:prior}) is equivalent to the prior for $\bm{\mu}\equiv\bm{\theta}\oslash\widehat{\bm{\sigma}}_{\bm{\theta}}$ based on a prior $\bm{\theta}\sim\mathcal{N}_m(\widehat{\bm{\theta}},\mathrm{diag}(\bm{\sigma}^2_{\widehat{\bm{\theta}}}))$ on the $m$ effect sizes $\bm{\theta}$, previously considered for $m=1$ effect size \citep{MicheloudHeld22}. The $\mathcal{N}_m(\bm{\mu}\mid\widehat{\bm{\theta}}\oslash\bm{\sigma}_{\widehat{\bm{\theta}}},\bm{I}_m)$ prior is convenient because it allows for prior mean specification without needing explicit knowledge of standard errors possibly unavailable from a given published study, as illustrated in the case study in \S\ref{Section:Illustration}. 

That said, publication bias and the winner’s curse often lead to overestimated original effect estimates \citep[][for example]{Ioannidis08,ButtonEtAl13,AndersonMaxwell17}, implying that for a power analysis of a replication study, the $\mathcal{N}_m(\bm{\mu}\mid\widehat{\bm{\theta}}\oslash\bm{\sigma}_{\widehat{\bm{\theta}}},\bm{I}_m)$ prior may be over-optimistic and lead to underpowered replication studies. A simple correction for this over-optimism is to multiply the effect sizes $\widehat{\bm{\theta}}$ in this prior by a vector factor $\bm{d}=(d_1,\ldots,d_m)^\top$ that takes on values between the vector of zeros $\bm{0}_m$ and a vector of ones $\bm{1}_m$, with shrinkage factor $\bm{s}=\bm{1}_m-\bm{d}$ that can be chosen based on previous replication studies in the same field \citep[as in][for the $m=1$ case]{MicheloudHeld22}. Further, a prior for $\bm{\mu}$ can be based on a (non-diagonal) s.p.d. covariance matrix $\Sigma_{\widehat{\theta}}$ for $\bm{\theta}$. More generally, a power prior distribution \citep{IbrahimEtAl15} can be specified for $(\bm{\mu})$ (or even for $(\bm{\mu},\bm{P},\bm{\nu})$), defined a posterior distribution constructed from previous historical data (likelihood) and a prior density for the parameters \citep{PawelEtAl24}.

Simulation Algorithm \ref{alg:MTPpower} calculates the predictive powers for either the B-MTP, H-MTP, BY-MTP, and/or DP-MTP, for $m$ given hypothesis tests, applicable to either an offline or online multiple testing scenario. This algorithm efficiently samples from the specified prior distribution for $\bm{\mu}$, and rapidly samples from the uniform prior distribution on $m\times m$ correlation matrices using the \citet{PourahmadiWang15} algorithm, run by the \texttt{randcorr} package \citep{MakalicEtAl18} of the \texttt{R} software \citep{RCoreTeam25}.

{Algorithm \ref{alg:MTPpower} employs a default ``objective" uniform prior that support all valid correlation matrices, for reasons explained in \S1. However, when necessary, the algorithm can be easily modified in obvious ways to multiple testing situations where the prior on the correlation matrix is informative, such as a prior supporting correlation matrices that have either an equi-correlated structure, or alternatively an AR(1) structure (e.g., for time-series or repeated measures designs).}

\newpage
\begin{algorithm}[!ht]
    \caption{Simulation algorithm to compute the predictive marginal powers for each MTP.}
    \label{alg:MTPpower}
    \centering
    \begin{tabular}{l}
    \textit{Inputs:} The $m\geq1$ test procedures used to test a given set of null hypotheses $\{H_{0,j}\}_{j=1}^m$, and their:\\
    \hspace{1.15cm} Test statistics $\bm{T}=(T_1,\ldots,T_m)^\top$; Type(s) of test(s) (lower-, upper, or two-tailed test);\\
    \hspace{1.15cm} Ratios, $(\widehat{\bm{\theta}}\oslash\bm{\sigma}_{\widehat{\bm{\theta}}})=(\widehat{\theta}_1/\sigma_{\widehat{\theta}_1},\ldots,\widehat{\theta}_m/\sigma_{\widehat{\theta}_m})^\top,$ of anticipated or estimated \\
    \hspace{1.15cm} effect sizes $\widehat{\bm{\theta}}=(\widehat{\theta}_1,\ldots,\widehat{\theta}_m)^\top$ divided by their corresponding standard errors; \\
    \hspace{1.15cm} Degrees of freedoms, $\bm{\nu}=(\nu_1,\ldots,\nu_m)^\top$; Total Type I error, $\alpha$; indicator function $\bm{1}(\cdot)$;\\
    \hspace{1.15cm} Choice of MTP to use: B-MTP, H-MTP, BY-MTP, and/or DP-MTP;\\
    \hspace{1.15cm} Number of samples of DP r.p.m.s for DP-MTP if used: $N$ (e.g., $N\equiv1000$);\\
    \hspace{1.15cm} Number of sampling iterations for power analysis, $S$ (e.g., $S\equiv5000$).\\
    \hline
    for $s=1,\ldots,S$ do: \\
    \qquad(a) Draw $\bm{\mu}^{(s)}\sim\mathcal{N}_m(\widehat{\bm{\theta}}\oslash\bm{\sigma}_{\widehat{\bm{\theta}}},\bm{I}_m)$ \hspace{.1cm} (or draw from another prior for $\bm{\mu}$ mentioned in paper).    \\
    \qquad(b) Draw $\bm{P}^{(s)}\sim\mathrm{LKJ}(\eta\equiv1)$\hspace{0.95cm}(using \citet{PourahmadiWang15} algorithm).\\
    \qquad(c) Draw $\bm{t}^{(s)}\sim\mathcal{NT}_m(\bm{\mu}^{(s)},\bm{P}^{(s)},\bm{\nu})$ \hspace{.6cm}(using stochastic representation eq.(\ref{eq:MNTdf})).\\
    \qquad(d) From each $t_j^{(s)}\in\bm{t}_m^{(s)}=(t_1^{(s)},\ldots,t_m^{(s)})^\top$, for $j=1,\ldots,m$, calculate (as intended) the \\
    \hspace{1.22cm} lower-, upper-, and/or 2-tailed $p$-values
    $p_1^{(s)}\equiv p(t_1^{(s)}),\ldots,p_j^{(s)}\equiv p (t_j^{(s)}),\ldots p_m^{(s)}\equiv p (t_m^{(s)})$.\\
    \qquad(e) Apply each used B-MTP, H-MTP, BY-MTP, and/or DP-MTP, on $\bm{p}^{(s)}$, at level $\alpha$, and\\
    \hspace{1.22cm} for sorted $p$-values $p^{(s)}_{(1)}\leq\ldots\leq p^{(s)}_{(r)}\leq\ldots\leq p^{(s)}_{(m)}$, calculate for $r=1,\ldots,m$,\\
    \hspace{1.22cm} $d_{\alpha,\text{M},(r)}=\bm{1}\{p_{(r)}\leq\Delta_{\alpha,\nu,\text{M}}(H_{0,(r)})\}$, for MTP $\text{M}\in\{\text{B},\text{H},\text{BY}\}$,\\
    \hspace{1.22cm} $d_{\alpha,\text{DP},(r)}^{(s)}\equiv\frac{1}{N}\sum_{g=1}^Nd_{\alpha,\text{DP},(r),g}^{(s)}$, where $d^{(s)}_{\alpha,\text{DP},(r),g}\equiv\bm{1}\{p^{(s)}_{(r)}\leq\Delta_{\alpha,\nu^{(s)}_g,\text{M}}(H_{0,(r)})\}$, \\
    \hspace{1.22cm} to yield: $(d_{\alpha,\text{B},(r)}^{(s)})_{r=1}^{m}$, $(d_{\alpha,\text{H},(r)}^{(s)})_{r=1}^{m}$, $(d_{\alpha,\text{BY},(r)}^{(s)})_{r=1}^{m}$, and/or $(d_{\alpha,\text{DP},(r)}^{(s)})_{r=1}^{m}$, as intended.\\
    end for
    \\
    \hline
    \textit{Output:} For each MTP used, $\text{M}\in\{\text{B},\text{H},\text{BY},\text{DP}\}$, the calculated estimates of:\\
    \hspace{1.23cm} \textbf{Predictive marginal powers}: \hspace{.13cm} $\overline{\textbf{pmp}}_{\alpha,\text{M}}\equiv\tfrac{1}{S}\sum\limits_{s=1}^S(d^{(s)}_{\alpha,\text{M},(r)})_{r=1}^m=(\bar{d}_{\alpha,\text{M},(r)})_{r=1}^m$;\\
    \hspace{1.23cm} \textbf{Predictive average power}: \hspace{.52cm} $\overline{\text{pap}}_{\alpha,\text{M}}\equiv\tfrac{1}{S}\tfrac{1}{m}\sum\limits_{s=1}^S\sum\limits_{r=1}^md^{(s)}_{\alpha,\text{M},(r)}$;\\
    \hspace{1.23cm} \textbf{Predictive disjunctive power}: \hspace{.01cm} $\overline{\text{pdp}}_{\alpha,\text{M}}\equiv\tfrac{1}{S}\sum\limits_{s=1}^S\bm{1}(\sum_{r=1}^md^{(s)}_{\alpha,\text{M},(r)}\geq1)$;\\
    \hspace{1.23cm} \textbf{Predictive conjunctive power}: \hspace{.02cm}$\overline{\text{pcp}}_{\alpha,\text{M}}\equiv\tfrac{1}{S}\sum\limits_{s=1}^S\bm{1}(\sum_{r=1}^md^{(s)}_{\alpha,\text{M},(r)}=m)$.\\
    \hspace{1.23cm} $p$\textbf{-value weights:} \hspace{2.1cm} $\widehat{\pi}_{\text{M}}(H_{0,(r)})\equiv\frac{\mathrm{exp}[\mathrm{log}(\bar{d}_{\alpha,\text{M},(r)})-\widehat{a}]}{\sum_{k=1}^m\mathrm{exp}[\mathrm{log}(\bar{d}_{\alpha,\text{M},(k)})-\widehat{a}]}$, for $r=1,\ldots,m$,\\
    \hspace{8cm} where $\widehat{a}\equiv \mathrm{max}_{k=1,\ldots,m}[\mathrm{log}(\bar{d}_{\alpha,\text{M},(k)})]$.\\
    \hspace{1.23cm} \textbf{Significance Chasing Bias (Hellinger Distance)}:\\ 
    \hspace{1.23cm} $\text{SigChase}_{(r)}\equiv{\frac{1}{\sqrt{2}}}[\{d^{1/2}_{\alpha,\text{M},(r)}(x) - \bar{d}^{1/2}_{\alpha,\text{M},(r)}\}^2 + \{(1-d_{\alpha,\text{M},j}(x))^{1/2} - (1-\bar{d}_{\alpha,\text{M},(r)})^{1/2}\}^2]$;\\
    \hspace{1.23cm} for sorted $p$-values $p_{(1)}\leq\ldots\leq p_{(r)}\leq\ldots\leq p_{(m)}$ and nulls $\{H_{0,(r)}\}_{r=1}^m$ tested on data $x$,\\
    \hspace{1.23cm} with $d_{\alpha,\text{M},(r)}(x)=\bm{1}\{p_{(r)}\leq\Delta_{\alpha,\nu,\text{M}}(H_{0,(r)})\}$, for $r=1,\ldots,m$, for MTP $\text{M}\in\{\text{B},\text{H},\text{BY}\}$,\\
    \hspace{1.23cm} and/or $d_{\alpha,\text{DP},(r)}\equiv\frac{1}{N}\sum_{g=1}^Nd_{\alpha,\text{DP},(r),g}=\frac{1}{N}\sum_{g=1}^N \bm{1}\{p_{(r)}\leq\Delta_{\alpha,\nu_g,\text{M}}(H_{0,(r)})\}$, as intended.\\
    \\
    \hspace{1.23cm} \textbf{Speed of Monte Carlo convergence of} $\bar{h}_{S,\text{M}}=\tfrac{1}{S}\sum_{s=1}^sh(\bm{t}_s)$:\\
    \hspace{5.7cm} Assessed by sample variance $v_{h,S}=\tfrac{1}{S^2}\sum_{s=1}^S\{h(\bm{t}_m^{(s)})-\bar{h}_S\}^2$,\\ 
    \hspace{5.7cm} for each $h\in\overline{\textbf{pmp}}_{\alpha,\text{M}}$ 
    and each $h\in\{\overline{\text{pap}}_{\alpha,\text{M}}, \overline{\text{pdp}}_{\alpha,\text{M}},\overline{\text{pcp}}_{\alpha,\text{M}}\}$.\\
    \end{tabular}
\end{algorithm}

\newpage
{Also, a typical hypothesis test is most often defined by a test statistic that has a known distribution under the null hypothesis, most often a normal or Student t distribution (at least approximately). This is why \S3-\S4 and the MTP power analysis Algorithm \ref{alg:MTPpower} emphasize the use of these distributions by default. In other situations, the hypothesis test statistic may have a non-normal, discrete, and/or unknown distribution under the null hypothesis (the latter which can be estimated through permutation- or bootstrap-based simulation), for one or more hypothesis test(s) among the $m$ total tests performed. In these other hypothesis testing situations, it is straightforward to use the non-normal and/or estimated distribution of the hypothesis test statistic under the null hypothesis, instead of the normal or Student $t$- null hypothesis distribution(s) within Algorithm \ref{alg:MTPpower} in order to conduct an MTP power analysis that would be more accurate than a power analysis assuming a null or Student $t$ null hypothesis distribution. Recall from \S2 that the super-uniformity condition does not make specific requirements on the distribution of the test statistic under the null hypothesis.}

\section{Case Study Illustration}\label{Section:Illustration}

The applicability of the new MTP predictive power analysis method is showcased through the analysis of  two-tailed $p$-values from forty-one hypothesis tests, presented in the third column of Table \ref{tab:Needleman}. 

These results were reported by a study \citep{NeedlemanEtAl79} of the neuropsychologic effects of childhood lead exposure measured from the shed baby teeth that teachers collected from first- and second-grade Massachusetts schoolchildren during 1975-1978. This landmark study applied an innovative non-invasive way to collect data on accumulated lead content in children (e.g., lead from ingested wall paint), and yielded statistical results that motivated much further research on lead poisoning and related areas, and inspired actions by the Center for Disease Control to lower the blood lead standard for children, and by the U.S. Congress and Environmental Protection Agency (EPA) to eliminate lead from gasoline, paint, plumbing, and other uses. Still, this study was criticized for methodological flaws, including that it reported $p$-values that did not jointly control for multiple comparisons for all reported $p$-values \citep[see pp.279-283][for further discussion]{WestfallYoung93}. Therefore, these $p$-values (subsets) have been reanalyzed several times by the MTP literature \citep[][Table 1]{BenjaminiYekutieli01}.

The teeth measurements identified a group of 100 children with low lead exposure ($<6$ ppm; lowest 10th percentile) and a group of 58 children with high lead exposure ($>24$ ppm; highest 10th percentile). Both groups were compared on each of 11 items of a teachers' behavior ratings scale, using respective ($\chi^2$) $z$-tests of equal group proportions of negative behavior responses. These groups were also compared using analysis of covariance (ANCOVA) $t$-tests, comparing group mean: total sum behavioral rating score; sub-scale and full-scale (or sum) scores on the Wechsler Intelligence Scales for Children measuring verbal IQ and performance IQ; and scores from Seashore and Token tests of verbal processing; and reaction times varied by four time intervals. Each ANCOVA controlled for $5$ covariates: mother’s age at subject’s birth; mother’s educational level; father’s socioeconomic status; number of pregnancies; and parental IQ. 

{In total, $m=41$ hypothesis tests were performed, which Table \ref{tab:Needleman} summarizes, including the number of tests performed for each kind of testing procedure and variable.} The fourth column of Table \ref{tab:Needleman} shows the test statistic ($t$) derived from each two-tailed $p$-value, obtained by $t_j = \Phi^{-1}(1-p_j/2)$ for each of the $z$-tests (i.e., tests $j=1,\ldots,11$ in Table \ref{tab:Needleman}), and by $t_j=\mathcal{T}^{-1}(1-p_j/2\mid0,1,\nu_j=151)$ for each of the ANCOVA $t$-tests (tests $j=12,\ldots,41=m$), where $\Phi$ denotes the standard normal $\mathcal{N}(0,1)$ CDF, and $\mathcal{T}$ is the standard $t$ CDF on $\nu_j=n_j-(k+2)=158-(5+2)=151$ degrees of freedom, based on $k=5$ covariates and 2 treatment groups. Each test statistic $t$ is an absolute ratio of an effect size measured by a raw group mean difference, divided by the standard error of this difference, where as appropriate, the effect size is either a between-group difference in proportions for a $z$-test, or a difference in dependent outcome means for an ANCOVA $t$-test.

The predictive MTP power analysis Algorithm \ref{alg:MTPpower} using $\alpha= 0.05$ was applied for $S=5000$ sampling iterations, to evaluate the powers of the DP-MTP (for $N\equiv1000$ iterations), the Bonferroni MTP, Holm MTP, and the Benjamini–Yekutieli MTP (BY-MTP), based on $41$-variate normal $\mathcal{N}_{41}(\widehat{\bm{\theta}}\oslash\bm{\sigma}_{\widehat{\bm{\theta}}},\bm{I}_{41})$ prior for $\bm{\mu}$, with mean vector $\widehat{\bm{\theta}}\oslash\bm{\sigma}_{\widehat{\bm{\theta}}}$ specified by the fourth column of Table \ref{tab:Needleman}, the vector of test statistics. Further, while a uniform LKJ prior is assigned to the correlation matrix of the 41 test statistics, the degrees of freedom for these test statistics were set by $\nu_j\rightarrow \infty$ for the $z$-tests of proportions indexed by $j=1,\ldots,11$ for each of the 11 behavior items, and set by $\nu_j\equiv151$ for each of the ANCOVA $t$-tests ($j=12,\ldots,41$). (For each ANCOVA $t$-test, \citet{NeedlemanEtAl79} did not provide $r^2$ of the combined 5 adjusting covariates nor the residual variance of the unadjusted dependent responses for both treatments, which could permit a more exact ANCOVA power analysis \citep{Shieh20}). Overall, for each of the four MTPs, this MTP predictive power analysis is for a future replication study applying the same 41 null hypothesis tests and total sample sizes $n_j=158$ used for each test, based on an ignorable shrinkage factor vector, $\bm{d}\equiv\bm{1}_m$ (a vector of $m$ ones) implying no over-optimism in the vector of $m$ effect sizes $\widehat{\bm{\theta}}$. The Supplementary Materials \citep{Karabatsos26supp} provides the software code used to run this MTP power analysis, including \texttt{R} code from the \texttt{bnpMTP} package \citep{Karabatsos25bnpMTP} and from the other cited packages. 

\begin{table}[!ht]
\centering
\begin{tabular}{ccccccccc}
  \hline
Test & Scale & $p$-value & $t$ stat. & MargPwr & $p$-weights & PrSig & PrSig.w & sigChase \\ 
  \hline
  1 & Behavior 1 			& 0.003 & 2.97 & 0.46 & 0.05 & 0.28 & 0.53$^\text{BY}$ & 0.13 \\ 
  2 & Behavior 2 			& 0.05 & 1.96 & 0.23 & 0.02 & 0.00 & 0.00 & 0.35 \\ 
  3 & Behavior 3 			& 0.05 & 1.96 & 0.23 & 0.02 & 0.00 & 0.00 & 0.35 \\ 
  4 & Behavior 4 			& 0.14 & 1.48 & 0.15 & 0.01 & 0.00 & 0.00 & 0.28 \\ 
  5 & Behavior 5 			& 0.08 & 1.75 & 0.18 & 0.02 & 0.00 & 0.00 & 0.31 \\ 
  6 & Behavior 6 			& 0.01 & 2.58 & 0.37 & 0.04 & 0.02 & 0.08$^\text{BY}$ & 0.35 \\ 
  7 & Behavior 7 			& 0.04 & 2.05 & 0.25 & 0.03 & 0.00 & 0.00 & 0.37 \\ 
  8 & Behavior 8 			& 0.01 & 2.58 & 0.38 & 0.04 & 0.02 & 0.08$^\text{BY}$ & 0.36 \\ 
  9 & Behavior 9 			& 0.05 & 1.96 & 0.23 & 0.02 & 0.00 & 0.00 & 0.35 \\ 
  10 & Behavior 10 		& 0.003 & 2.97 & 0.47 & 0.05 & 0.28 & 0.53$^\text{BY}$ & 0.14 \\ 
  11 & Behavior 11 		& 0.003 & 2.97 & 0.47 & 0.05 & 0.28 & 0.53$^\text{BY}$ & 0.14 \\ 
  12 & Sum Behavior 		& 0.02 & 2.35 & 0.30 & 0.03 & 0.00 & 0.00 & 0.40 \\ 
  13 & Verbal IQ 1 		& 0.04 & 2.07 & 0.25 & 0.03 & 0.00 & 0.00 & 0.37 \\ 
  14 & Verbal IQ 2 		& 0.05 & 1.98 & 0.23 & 0.02 & 0.00 & 0.00 & 0.35 \\ 
  15 & Verbal IQ 3 		& 0.02 & 2.35 & 0.31 & 0.03 & 0.00 & 0.00 & 0.41 \\ 
  16 & Verbal IQ 4 		& 0.49 & 0.69 & 0.06 & 0.01 & 0.00 & 0.00 & 0.18 \\ 
  17 & Verbal IQ 5 		& 0.08 & 1.76 & 0.18 & 0.02 & 0.00 & 0.00 & 0.30 \\ 
  18 & Verbal IQ 6 		& 0.36 & 0.92 & 0.08 & 0.01 & 0.00 & 0.00 & 0.20 \\ 
  19 & Performance IQ 1 	& 0.03 & 2.19 & 0.28 & 0.03 & 0.00 & 0.00 & 0.39 \\ 
  20 & Performance IQ 2 	& 0.38 & 0.88 & 0.07 & 0.01 & 0.00 & 0.00 & 0.19 \\ 
  21 & Performance IQ 3 	& 0.15 & 1.45 & 0.14 & 0.01 & 0.00 & 0.00 & 0.27 \\ 
  22 & Performance IQ 4 	& 0.54 & 0.61 & 0.05 & 0.01 & 0.00 & 0.00 & 0.17 \\ 
  23 & Performance IQ 5 	& 0.90 & 0.13 & 0.03 & 0.00 & 0.00 & 0.00 & 0.13 \\ 
  24 & Performance IQ 6 	& 0.37 & 0.90 & 0.08 & 0.01 & 0.00 & 0.00 & 0.20 \\ 
  25 & Full Verbal IQ 	& 0.03 & 2.19 & 0.27 & 0.03 & 0.00 & 0.00 & 0.38 \\ 
  26 & Full Perf. IQ 		& 0.03 & 2.19 & 0.28 & 0.03 & 0.00 & 0.00 & 0.39 \\ 
  27 & Full VerbalPerf.IQ & 0.08 & 1.76 & 0.20 & 0.02 & 0.00 & 0.00 & 0.32 \\ 
  28 & Seashore 1 		& 0.002 & 3.15 & 0.50 & 0.05 & 0.38 & 0.70$^\text{B,BY}$ & 0.09 \\ 
  29 & Seashore 2 		& 0.03 & 2.19 & 0.26 & 0.03 & 0.00 & 0.00 & 0.37 \\ 
  30 & Seashore 3 		& 0.07 & 1.82 & 0.19 & 0.02 & 0.00 & 0.00 & 0.32 \\ 
  31 & Total Seashore 	& 0.002 & 3.15 & 0.51 & 0.05 & 0.38 & 0.70$^\text{B,BY}$ & 0.09 \\ 
  32 & Token 1 			& 0.37 & 0.90 & 0.07 & 0.01 & 0.00 & 0.00 & 0.19 \\ 
  33 & Token 2 			& 0.90 & 0.13 & 0.04 & 0.00 & 0.00 & 0.00 & 0.14 \\ 
  34 & Token 3 			& 0.42 & 0.81 & 0.07 & 0.01 & 0.00 & 0.00 & 0.19 \\ 
  35 & Token 4 			& 0.05 & 1.98 & 0.21 & 0.02 & 0.00 & 0.00 & 0.34 \\ 
  36 & Total Token 		& 0.09 & 1.71 & 0.18 & 0.02 & 0.00 & 0.00 & 0.30 \\ 
  37 & Sentence 			& 0.04 & 2.07 & 0.24 & 0.02 & 0.00 & 0.00 & 0.36 \\ 
  38 & Reaction Time 1 	& 0.32 & 1.00 & 0.08 & 0.01 & 0.00 & 0.00 & 0.20 \\ 
  39 & Reaction Time 2 	& 0.001 & 3.36 & 0.56 & 0.06 & 0.57$^\text{B,H}$ & 0.76$^\text{B,H,BY}$ & 0.01 \\ 
  40 & Reaction Time 3 	& 0.001 & 3.36 & 0.55 & 0.05 & 0.57$^\text{B,H}$ & 0.76$^\text{B,H,BY}$ & 0.02 \\ 
  41 & Reaction Time 4 	& 0.01 & 2.61 & 0.37 & 0.04 & 0.02 & 0.08$^\text{BY}$ & 0.35 \\ 
   \hline
\end{tabular}
\caption{\small Two-tailed $p$-values from \citet[][Tables 3, 7-8]{NeedlemanEtAl79}; test ($t$) statistics; and DP-MTP results of marginal predictive power (MargPwr); $p$-value weights; probability of discovery for each $p$-value, PrSig (each weighted $p$-value PrSig.w); and significance chasing bias index in Hellinger distance. Superscripts are significant discoveries according to un/weighted B-MTP, H-MTP, or BY-MTP.} 
\label{tab:Needleman}
\end{table}

The following results of the predictive MTP power analysis provided by Algorithm \ref{alg:MTPpower} are as follows, for the DP-MTP, B-MTP, H-MTP, and BH-MTP. The DP-MTP can be considered as the baseline MTP method, because as mentioned earlier, the DP prior supports the B-MTP, H-MTP, and BH-MTP, and all other MTPs which control the FWER or FDR under arbitrarily-correlated $p$-values. For the 41 hypothesis tests, the DP-MTP, B-MTP, H-MTP, and BH-MTP obtained average power 0.25, 0.21, 0.22, and 0.26, respectively; all with disjunctive power 1.00 and conjunctive power 0. 

{Table \ref{tab:Needleman} shows the marginal power for each of the 41 hypothesis tests, estimated by Algorithm \ref{alg:MTPpower}. These are respectively the marginal powers of the 41 test statistics, with respect to the prior predictive distribution of the test statistics under their alternative hypotheses, under the prior centered on the observed effect sizes, and the uniform prior for the correlation matrices. The table shows that the marginal power tends to increase with the absolute $t$-statistic, a function of the absolute effect size, as intuitively expected. In other words, these are the marginal powers of the 41 hypothesis tests after accounting for arbitrarily-correlated $p$-values (test statistics), powers which are different across these tests. The marginal powers are not very large for the tests, which reflects the relatively small sample size of the study and confirms the known fact that there is a decrease in power for individual tests in multiple testing scenarios \citep{SennBretz07}. Further,} the DP-MTP marginal powers (Table \ref{tab:Needleman}) are similar to those of the three other MTPs (Figure \ref{fig:MTP_Scatterplots}, left side). 

For DP-MTP, Table \ref{tab:Needleman} also shows the $p$-value weights, and probability of significance discovery (PrSig and PrSig.w) for each unweighted and weighted $p$-value among the 41 total tests. These results are compared with the significance discovery indicators for each of the unweighted and weighted versions of B-MTP, H-MTP, and BY-MTP, obtained from the \texttt{p.adjust()} and \texttt{p.adjust.w()} codes of \texttt{R} package $\texttt{someMTP}$ \citep{Finos21}. DP-MTP indicates that any test results in a discovery by PrSig $>$ 0 (or by PrSig.w $>$ 0 for weighted testing), while providing uncertainty quantification in multiple testing. The Table shows that zero values of PrSig (and PrSig.w) correspond to non-discoveries found by the other MTPs.

For DP-MTP, Table \ref{tab:Needleman} shows the significance chasing bias index for each of the 41 $p$-values (tests), where non-bias is indicated by a small index value. The 41 bias indices of DP-MTP tended to be lower than the corresponding indices of B-MTP, H-MTP, and BY-MTP (Figure \ref{fig:MTP_Scatterplots}, right side), perhaps because of the DP prior supports all MTPs controlling the FWER or FDR instead of one significance cutoff for $p$-values.

{We now consider a prior sensitivity analysis for the DP-MTP, with respect to four different levels of the shrinkage factor (\S3). The shrinkage (and corresponding sensitivity analysis) can address the fact that, for a future
replication setting, the previously observed test statistics defining the prior mean can possibly be affected by overestimation and the winner’s curse, and so the shrinkage can help rule out any systematic optimism in the reported power estimates.}

\begin{figure}[!ht]
    \centering
    \includegraphics[width=1\linewidth]{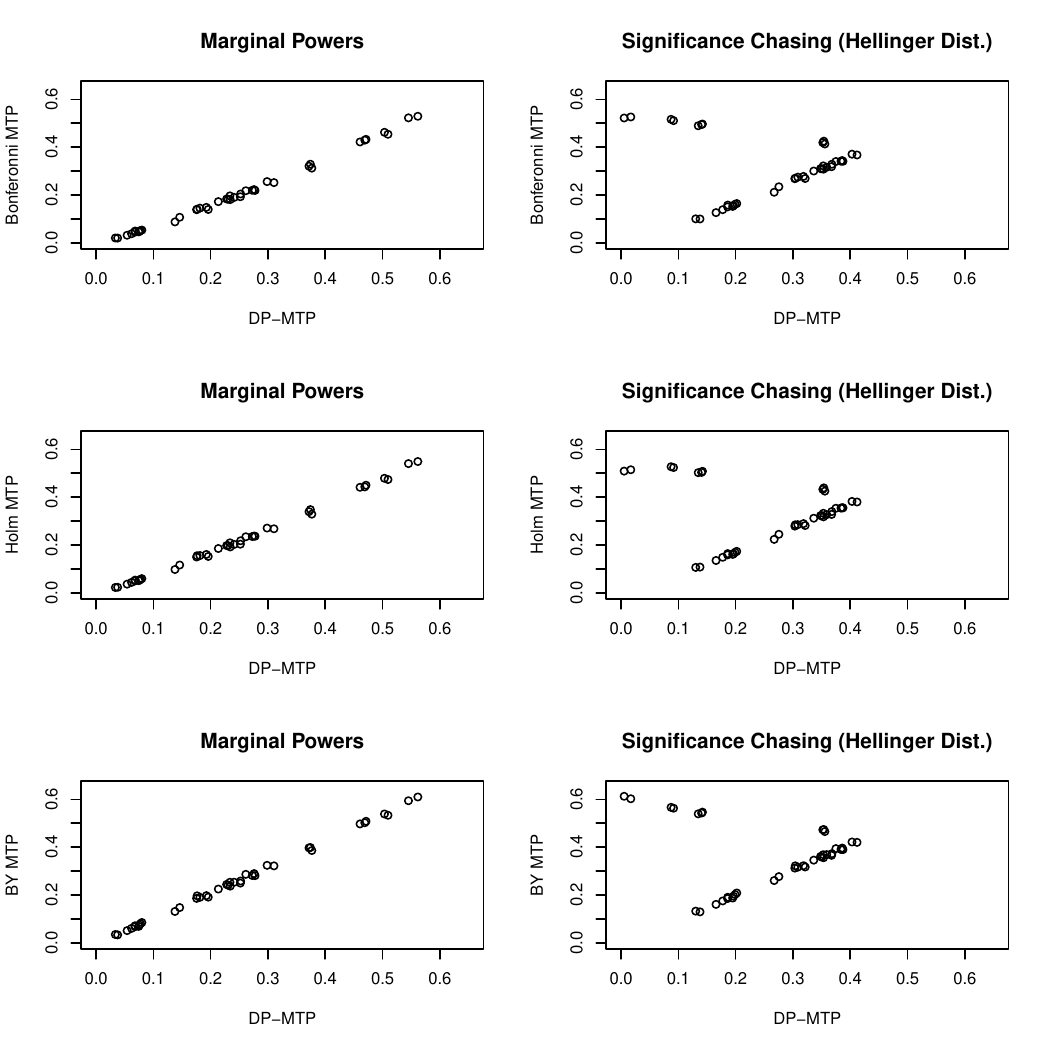}
    \caption{For the 41 tests, comparing marginal powers and significance chasing biases between 4 MTPs.}
    \label{fig:MTP_Scatterplots}
\end{figure}

\begin{figure}[!ht]
    \centering
    \includegraphics[width=1\linewidth]{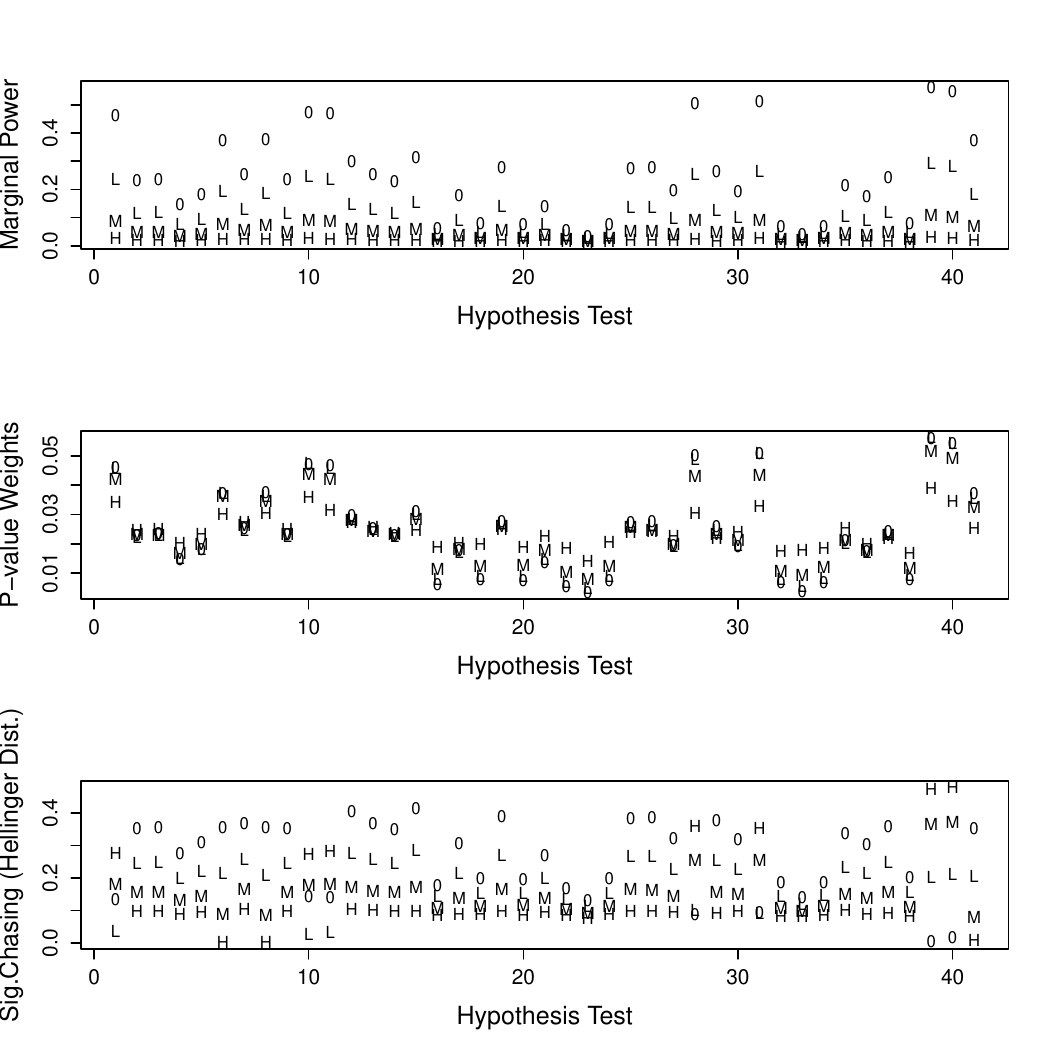}
    \caption{For the 41 tests analyzed by the DP-MTP, comparing marginal powers, $p$-value weights, and significance chasing biases across four levels of shrinkage, namely: Zero-shrinkage (0), low shrinkage of $s=1/4$ (L), medium shrinkage of $s=1/2$ (M), and high shrinkage of $s=3/4$ (H).}
    \label{fig:MTP_Scatterplots_ShrinkSens}
\end{figure}

%d = rep(1,m ) # No  shrinkage: s = 1 - d = rep(  0  , m )
%d = (3/4)*rep( 1,m) # Low shrinkage: s = 1 - d = rep( 1/4 , m )
%d = (2/4)*rep( 1,m) # Med shrinkage:s = 1 - d = rep( 1/2 , m )
%d = (1/4)*rep( 1,m) # Hi  shrinkage:s = 1 - d = rep( 3/4 , m )

{Specifically, we consider shrinkage factors $s=0$ (already considered above), $1/4$, $1/2$, and $3/4$, such that each shrinkage factor $s$ is defined as the same for all 41 hypothesis tests, and they are respectively referred to as zero, low, medium and large shrinkage factors. For these shrinkage factors, the DP-MTP respectively obtained: average powers of 0.25, 0.13, 0.05, and 0.02; disjunctive powers of 1.00, 0.98, 0.69, 0.29; and conjunctive powers of zero.}

{Figure \ref{fig:MTP_Scatterplots_ShrinkSens} compares the marginal powers, $p$-value weights, and significance-chasing Hellinger distance measures across these four shrinkage levels. Overall, as expected, the measures of power decreases with increasing shrinkage level. Interestingly, the estimated $p$-value weights are rather insensitive to varying shrinkage. The significance-chasing measures mostly increased with decreasing shrinking level, but not always, and not necessarily in a strictly-ordered fashion. As reasonably expected, these measures can vary with respect to the specified prior distribution on the effect sizes.}

\section{Theorem}\label{Section:Theorem} 

Algorithm \ref{alg:MTPpower} addresses MTP predictive power analysis problem by evaluating the expectation:
\begin{equation}\label{eq:MCintegral}
\mathbb{E}_f\{h(\bm{T})\}=\int\dots\int_\mathcal{T}h(\bm{t}_m)f(\bm{t}_m)\text{d}t_1,\ldots,\text{d}t_m,
\end{equation}
where $f$ is the marginal prior predictive probability density function of test statistics $\bm{t}_m=(t_1,\ldots,t_m)^\top$: 
\begin{equation}\label{eq:marginal_pdf}
f(\bm{t}_m)\equiv\int\dots\int\mathcal{NT}_m(\bm{t}_m\mid\bm{\mu},\bm{P},\bm{\nu})\mathrm{LKJ}(\bm{P}\mid\eta\equiv1)\text{d}\mu_1\cdots\text{d}\mu_m\text{d}\rho_{1,2}\cdots\text{d}\rho_{m-1,m},
\end{equation}
and where $h:\mathbb{R}^m\rightarrow\mathbb{R^+}$ is a bounded function $h(\cdot)$ of $\bm{t}_m$ corresponding to $m$ ordered $p$-values $\underline{\bm{p}}(\bm{t}_m)\equiv(p_{(r)}\equiv p(t_{(r)}))_{r=1}^m$, and defined either for a vector of predictive marginal powers: $\textbf{pmp}_{\alpha,\text{M}}=(h\{p(t_{(r)})\}\equiv d_{\alpha,\text{M},(r)})_{r=1}^m\equiv(\bm{1}\{p_{(r)}\leq\Delta_{\alpha,\nu,\text{M}}(H_{0,(r)})\})_{r=1}^m$; for predictive average power: $h\{\underline{\bm{p}}(\bm{t}_m)\}\equiv\text{pap}=\frac{1}{m}\sum^m_{r=1}d_{\alpha,\text{M},(r)}(x)$; for disjunctive power: $h\{\underline{\bm{p}}(\bm{t}_m)\}\equiv\text{pdp}=\bm{1}(\sum_{r=1}^md_{\alpha,\text{M},{(r)}}\ge1)$; or for conjunctive power: $h\{\underline{\bm{p}}(\bm{t}_m)\}\equiv\text{pcp}=\bm{1}(\sum_{r=1}^md_{\alpha,\text{M},{(r)}}=m)$.

For any choice of these functions $h$, the power analysis algorithm estimates $\mathbb{E}_f\{h(\bm{T})\}$ by $\bar{h}_S=\tfrac{1}{S}\sum_{s=1}^Sh(\bm{t_m}^{(s)})$, and assesses the speed of convergence of $\bar{h}_S$ by the sample variance:
\begin{equation}\label{eq:MCvar}
v_{h,S}=\tfrac{1}{S^2}\sum_{s=1}^S\{h(\bm{t_m}^{(s)})-\bar{h}_S\}^2.
\end{equation}

\begin{theorem}
The estimator $\bar{h}_S$ almost surely (a.s.) converges $\bar{h}_S\overset{a.s.}{\rightarrow}\mathbb{E}_f\{h(\bm{T})\}$ by the Strong Law of Large Numbers (SLLN), and $v_S$ is an estimate of the variance:
\begin{equation}\label{eq:varMCintegral}
\mathbb{V}\{\bar{h}_S\}=\frac{1}{S}\int\dots\int_\mathcal{T}[h(\bm{t}_m)-\mathbb{E}_f\{h(\bm{T})\}]^2f(\bm{t}_m)\mathrm{d}t_1,\ldots,\mathrm{d}t_m.
\end{equation}
\end{theorem}
\begin{proof}
For the iid $h(\bm{t_m}^{(s)})$ obtained from $\bm{t_m}^{(s)}\overset{iid}{\sim}f$, $h<\infty$ implies $\mathbb{E}_f\{h(\bm{T})\}<\infty$ with finite second moment. Then, a.s. convergence of $\bar{h}_S$ under SLLN holds \citep[][\S22]{Billingsley95}, and $h^2$ has a finite expectation under $f$, which ensures that $v_S$ estimates (\ref{eq:varMCintegral}).
\end{proof}

Further, given that each of the above choices of function $h$ is binary (0 or 1) valued, its variance attains maximum $1/4$ when $\mathbb{E}_f\{h(\bm{T})\}=1/2$, which implies the upper bound $1/4S$ for its Monte Carlo variance (\ref{eq:MCvar}). For example, the case study in \S\ref{Section:Illustration} was based on running $S\equiv5000$ sampling iterations of the algorithm, yielding an upper bound of $1/4S=1/20000=0.00005$ for the Monte Carlo variance of each of the presented estimates of power, implying fast convergence of the Monte Carlo power analysis algorithm.

\section{Conclusions}\label{Section:ConclusionsDiscussion} 

This study proposed and illustrated a practical automatic method for congenially evaluating the power of any MTP which controls FWER or FDR under unknown and un-analyzed arbitrary dependencies (inter-correlations) between $p$-values. The predictive power analysis method is defined by a general prior distribution for the effect sizes and a joint uniform prior for the correlation matrix for the test statistics to fully account for uncertainty in these parameters. 

The new method not only can be used to marginal powers of tests (resp.) or corresponding sample size determinations (given desired marginal powers) for a future planned (e.g. replication or interim) study, but also the calculated marginal powers can either be used to weight $p$-values to increase MTP power while minimizing the relative impacts of any significance chasing biases, or compared with raw $p$-values to evaluate for the presence of such biases. 

{As mentioned, Algorithm \ref{alg:MTPpower} outputs consistent estimates of the marginal predictive powers of the given $m$ hypothesis tests, which can then be used as a basis for assigning weights to $p$-values, respectively, in a subsequent application of a weighted MTP based on these $p$-value weights. Recall that a weighted MTP, which assigns relatively higher weights to the subset of tested null hypotheses that are likely to be false, tends to have higher statistical power compared to unweighted multiple testing procedure. An interesting open question is how accurately can such an MTP detect the subset of truly false null hypotheses in any given multiple hypothesis testing problem, for a range of plausible correlation structures for the test statistics ($p$-values). This question deserves to be addressed extensively in future research.}

\section*{Acknowledgments}
This research is supported in part by National Institute for Health grant 1R01AA028483-01. {The author declares no conflict of interest, and gives special thanks to the journal editors and anonymous reviewers for editorial suggestions which have helped improve the presentation of this paper}, first made publicly available as an \texttt{arXiv} \texttt{arXiv:2603.07312} on March 7, 2026. Section \S2 of this paper was presented at the 14th International Conference on Bayesian Nonparametrics (BNP14) at UCLA on June 26, 2025. 

{
\appendix
\begin{center}
\Large\textbf{Appendix: A Review of Existing Algorithms for Optimal $p$-value Weighting, and for Calculating the Correlation Matrix of Test Statistics}
\end{center}

For Bonferroni MTPs, optimal $p$-value weights can be found using any one of various optimization algorithms that maximize either the: weighted power \citep{WestfallKrishen01}; expected number of rejections \citep{ZhangEtAl15}; average power based on explicit formulae \citep{RubinEtAl06,RoederWasserman09}; and disjunctive or conjunctive power, while maximizing the asymptotic $m$-variate normal distribution of $z$-scores of one-sided $p$-values for hypothesis tests, given their respective specified levels of marginal (conditional) powers \citep{XiChen24}. For graphical or chain MTPs, optimal $p$-value weights can be found using grid search or simulation algorithms \citep{BretzEtAl11pwr,WiensEtAl13,DmitrienkoEtAl15} or deep learning algorithms maximizing weighted power \citep{ZhanEtAl22}.

For specialized multiple testing problems employing $m$ hypothesis tests, the $m\times m$ correlation matrix of the test statistics can be derived either: for uncorrelated test statistics ($p$-values), perhaps from an orthogonal design or contrasts; or when their covariance matrix or joint distribution can be estimated, perhaps using permutation or bootstrap sampling methods via the \citet{WestfallYoung93} MTP or related MTPs; while such a covariance matrix estimate can be used to whiten-transform \citep{KessyEtAl18} the original test statistics into $m$ independent asymptotic normal $\mathcal{N}(0,1)$ test statistics; or for pairwise mean comparisons done via a Dunnett- or a Tukey-Kramer-type MTP. The \citet{Dunnett55} MTP, which compares the outcome mean between each of $m$ given treatment groups to the outcome mean of the control group, the correlation matrix can be calculated as a simple square-root function of group sample sizes, such that, for example, a balanced design produces a $m\times m$ matrix with all distinct pairwise correlations equaling $0.50$. Then the resulting $m$ pairwise mean comparison test statistics under the null hypotheses follows a $m$-variate $t$-distribution with zero-location(mean) vector and with scale matrix determined by these correlations and a common degrees of freedom \citep[][pp.1101-1103]{Dunnett55}. For the \citet{Tukey49}-\citet{Kramer56} MTP, which compares means between the distinct pairs of all treatment groups, explicit formulas for the correlation matrix of $m$ test statistics can be derived for certain, e.g., balanced incomplete block, ANCOVA, incomplete Block Lattice, or partially balanced designs; \citep{Kramer57}. But the results of this MTP can conflict with the result of the $F$-test of equality of means between all three or more treatment groups \citep{GurvichNaumova21}.}

\newpage
\bibliographystyle{apalike2}
\bibliography{Karabatsos.bib}

\end{document}